\documentclass[
    11pt,
    a4paper,
    oneside,
    onecolumn,
    unpublished,
    final
]{quantumarticle}

\pdfoutput=1
\usepackage{microtype}
\usepackage[utf8]{inputenc}
\usepackage[english]{babel}
\usepackage[numbers,sort&compress]{natbib}

\usepackage{xcolor}
\definecolor{cblue}{rgb}{0.16, 0.32, 0.75}
\definecolor{cred}{rgb}{0.7, 0.11, 0.11}

\usepackage[
	    colorlinks=true, 
	    breaklinks=false, 
	    linkcolor=blue, 
		draft=false,
	    bookmarksopen=false, 
	    bookmarksopenlevel=0,
	    bookmarksnumbered=true,
]{hyperref}
\hypersetup{%
	,linkcolor=cred
	,citecolor=cblue
	,urlcolor=black
}

\usepackage{
    amsmath,
    amssymb, 
    amsfonts,
    amsthm,
    mathtools,
}
\allowdisplaybreaks

\usepackage{
    bbold,
    dsfont,
    nicefrac,
    derivative, 
    tensor,
    colonequals
}

\usepackage{todonotes}
\usepackage{etoolbox}   
\usepackage{xpatch}
\usepackage{ifthen}
\usepackage{ifdraft}
\usepackage{cleveref} 
\usepackage{csquotes}   

\usepackage{xcolor}
\usepackage{graphicx}
\usepackage{svg}
\usepackage{caption}
\usepackage{subcaption}
\usepackage{booktabs}
\usepackage{footnote}

\newcommand\mrm\mathrm

\theoremstyle{plain}
\newtheorem{theorem}{Theorem}

\newtheorem{lemma}{Lemma}

\newtheorem{corollary}{Corollary}

\newtheorem{prop}{Proposition}

\theoremstyle{definition}

\theoremstyle{remark}
\newtheorem{remark}{Remark}

\DeclareMathOperator{\one}{\mathds{1}}
\newcommand{\denseD}{\mathcal{D}_{0}}
\newcommand{\fockspace}{\mathcal{F}}

\newcommand\finite{\mathrm{fin}}

\newcommand\finitePhotonFock{\fockspace_{\finite}}
\newcommand\finiteFock\finitePhotonFock

\newcommand\creation{a^\dagger}
\newcommand\ad\creation
\newcommand\adag\creation
\newcommand\annihilation{a}
\newcommand\an\annihilation
\newcommand\photonnumber{\hat{n}}
\newcommand\nn\photonnumber

\newcommand{\numberOfExcitations}{\mathcal N}
\newcommand\numExc\numberOfExcitations
\newcommand{\atomNumber}{N}
\newcommand{\atomNumberUpState}{\atomNumber{}_+}
\newcommand\Nup\atomNumberUpState
\newcommand{\atomNumberDownState}{\atomNumber{}_-}
\newcommand\Ndown\atomNumberDownState

\newcommand\spin[1]{S^{#1}}
\newcommand\s{\spin} 

\newcommand\Splus{\spin+}
\newcommand\Sminus{\spin-}
\newcommand\Sp\Splus
\newcommand\Sm\Sminus
\newcommand\SX{S_x}
\newcommand\Sx{\SX}
\newcommand\SY{S_y}
\newcommand\Sy{\SY}
\newcommand\SZ{S_z}
\newcommand\Sz{\SZ}

\newcommand\photonnumberalt{\boldsymbol{\nn}}
\newcommand\nnalt\photonnumberalt
\newcommand\creationalt{\boldsymbol{a^\dagger}}
\newcommand\adalt\creationalt
\newcommand\adagalt\creationalt
\newcommand\annihilationalt{\boldsymbol{a}}
\newcommand\analt\annihilationalt

\newcommand\spinalt[1]{\boldsymbol{S^{#1}}}
\newcommand\salt\spinalt

\newcommand\SXalt{\boldsymbol{S_x}}
\newcommand\Sxalt\SXalt
\newcommand\SYalt{\boldsymbol{S_y}}
\newcommand\Syalt\SYalt
\newcommand\SZalt{\boldsymbol{S_z}}
\newcommand\Szalt\SZalt

\newcommand\pauli[1]{\sigma_{#1}}
\newcommand\pauliI[2]{\pauli{#1}^{(#2)}} 
\newcommand\paulii\pauliI

\newcommand\pauliPlus{\pauli{+}}
\newcommand\pauliMinus{\pauli{-}}

\newcommand\pauliPlusI[1]{\pauliI{+}{#1}}
\newcommand\pauliMinusI[1]{\pauliI{-}{#1}}

\newcommand{\tcpar}[2]{z_{#1}^{#2}}

\newcommand{\tcparvec}[1]{\underbar{z}^{#1}}
\newcommand{\tcpsi}[3]{\psi_{#1,#2}^{#3}}
\newcommand{\tcpsinormed}[3]{\tilde{\psi}_{#1,#2}^{#3}}
\newcommand{\tceig}[3]{E_{#1,#2}^{#3}}

\newcommand{\M}{\boldsymbol{\hat M}}

\newcommand{\e}{\mathrm{e}}
\newcommand{\iu}{\mathrm{i}\mkern1mu}

\let\Re\undefined
\DeclareMathOperator{\Re}{Re}
\let\Im\undefined
\DeclareMathOperator{\Im}{Im}

\newcommand{\dom}{\mathrm{Dom}}

\newcommand{\dd}[1]{\mathop{\mathrm{d}#1}}

\newcommand{\adj}{^\ast}

\newcommand\close\overline
\newcommand{\primed}{^\prime}

\newcommand{\HH}{\mathcal{H}}
\newcommand{\N}{\mathds{N}} 
\newcommand{\R}{\mathds{R}} 
\newcommand{\C}{\mathds{C}} 

\newcommand\Ltwo[1]{L^2(#1)}

\DeclarePairedDelimiter\of()
\DeclarePairedDelimiter\cof{\{}{\}}
\DeclarePairedDelimiter\eof{[}{]}

\DeclarePairedDelimiterX\comm[2]{[}{]}{
    \ifblank{#1}{\;\cdot\;}{#1}, \ifblank{#2}{\;\cdot\;}{#2}
}
\newcommand\commutator\comm

\DeclarePairedDelimiterX\innerp[2]{\langle}{\rangle}{
    \ifblank{#1}{\;\cdot\;}{#1},\; \ifblank{#2}{\;\cdot\;}{#2}
}

\DeclarePairedDelimiterX\braket[2]{\langle}{\rangle}{
    \ifblank{#1}{\;\cdot\;}{#1}\,\delimsize\vert\, \ifblank{#2}{\;\cdot\;}{#2}
}

\DeclarePairedDelimiterX\ket[1]{\lvert}{\rangle}{#1}

\newcommand{\emptyplaceholder}{\;\cdot\;} 
\DeclarePairedDelimiterX\norm[1]\lVert\rVert{
    \ifblank{#1}{\emptyplaceholder}{#1}
}

\DeclarePairedDelimiterXPP\pnorm[1]{}\lVert\rVert{_{p}}{
    \ifblank{#1}{\emptyplaceholder}{#1}
}

\DeclarePairedDelimiterXPP\twonorm[1]{}\lVert\rVert{_{2}}{
    \ifblank{#1}{\emptyplaceholder}{#1}
}

\DeclarePairedDelimiterXPP\opnorm[1]{}\lVert\rVert{_{\mathrm{op}}}{
    \ifblank{#1}{\emptyplaceholder}{#1}
}

\DeclarePairedDelimiterX\abs[1]\lvert\rvert{#1}
\DeclarePairedDelimiterX\aof[1]\langle\rangle{#1}
\providecommand\given{}
\DeclarePairedDelimiterXPP\Prob[1]{\mathcal{P}}(){}{
   \renewcommand\given{\nonscript\:\delimsize\vert\nonscript\:\mathopen{}}
   #1}

\providecommand\given{}
\newcommand\SetSymbol[1][]{%
   \nonscript\:#1\vert
   \allowbreak
   \nonscript\:
   \mathopen{}}
\DeclarePairedDelimiterX\Set[1]\{\}{%
\renewcommand\given{\SetSymbol[\delimsize]}
#1
}

\hypersetup{
	    pdftitle={RWA for Dicke},
	    pdfauthor={Leonhard Richter, Daniel Burgarth, Davide Lonigro},
	}

\title{Quantifying the rotating-wave approximation of the Dicke model}

\author{Leonhard Richter}
\orcid{0009-0003-0482-6133}
\email{leonhard.richter@fau.de}

\author{Daniel Burgarth}
\orcid{0000-0003-4063-1264}

\author{Davide Lonigro}
\orcid{0000-0002-0792-8122}
\affiliation{Department Physik, Friedrich-Alexander-Universität Erlangen-Nürnberg, Staudtstra{\ss}e 7, 91058 Erlangen, Germany}
\date{07.08.2025}

\begin{document}

\maketitle

\begin{abstract}\sloppy 
We analytically find quantitative, non-perturbative bounds to the validity of the rotating-wave approximation (RWA) for the multi-atom generalization of the quantum Rabi model: the Dicke model. 
Precisely, we bound the norm of the difference between the evolutions of states generated by the Dicke model and its rotating-wave approximated counterpart, that is, the Tavis--Cummings model. 
The intricate role of the parameters of the model in determining the bounds is discussed and compared with numerical results. 
Our bounds are intrinsically state-dependent and, in particular, capture a nontrivial dependence on the total angular momentum of the initial state; this behaviour also seems to be confirmed by accompanying numerical results. 
\end{abstract}

\section{Introduction}

The rotating-wave approximation (RWA), in its numerous and diverse incarnations, consists in neglecting or averaging out fastly oscillating terms in a differential equation---the latter being, for closed quantum systems, the Schr\"odinger equation. More often than not, such a procedure yields simpler models for which analytical or semi-analytical methods are available~\cite{Rabi_Zacharias_Millman_Kusch_1938,Rabi_1937,Garraway_2011,heibBoundingRotatingWave2025,Tavis_Cummings_1968,Bogoliubov_Bullough_Timonen_1996,larson2021jaynes}.

In particular, the RWA is an extremely common tool in the study of the interaction between quantum mechanical matter and boson fields, the benchmark case being the quantum Rabi model---the quantum counterpart of the Rabi model, developed by Jaynes and Cummings in the 1960s, which describes the interaction between a single two-level system (a spin) and a monochromatic boson field~\cite{Jaynes_Cummings_1963}. The Schrödinger equation corresponding to this model is generated by the unbounded linear operator given by the following expression (all tensor products being implicit):
\begin{equation}
    H_{\mrm{Rabi}}=\frac{\omega_0}{2}\sigma_z+\omega a^\dag a+\lambda(\sigma_++\sigma_-)(a+a^\dag),
\end{equation}
with $\omega_0,\omega$ respectively being the difference between the energies of the two spin levels, and the frequency of the field; $a,a^\dag$ being the annihilation and creation operators of the field; and $\sigma_\pm$ being the ladder operators implementing transitions between the two spin levels. Here, the RWA consists in replacing the model above with the following one: 
\begin{equation}
    H_{\mrm{JC}}=\frac{\omega_0}{2}\sigma_z+\omega a^\dag a+\lambda(\sigma_+a+\sigma_-a^\dag),
\end{equation}
which is usually denoted as the Jaynes--Cummings model; that is, the terms $\sigma_+a^\dag$ and $\sigma_-a$ are suppressed. The resulting model is far easier to solve since it exhibits a symmetry---conservation of the energy of the non-interacting system---which would be otherwise broken by the neglected terms~\cite{Jaynes_Cummings_1963}.

The heuristic justification of the RWA in the quantum Rabi model goes as follows. Assuming the resonance condition $\omega=\omega_0$ for simplicity, in the interaction picture corresponding to the free Hamiltonian $H_0=\omega\sigma_z/2+\omega a^\dag a$ the interaction terms $\sigma_+a^\dag$ and $\sigma_-a$ acquire a time-dependent oscillatory term $\mathrm{e}^{\pm2\mathrm{i}\omega t}$.
In spite of some hand-waving arguments for these terms to vanish in the limit of large frequencies \(\omega \to \infty\) being known for long time~\cite{Jaynes_Cummings_1963,klimovGrouptheoreticalApproachQuantum2009}
and of the ubiquity of the RWA in applications, a rigorous justification of the RWA for the quantum Rabi model---together with a quantitative evaluation of the error committed when performing the RWA for finite values of the parameters---is relatively recent.
This was obtained in Ref.~\cite{burgarthTamingRotatingWave2024} (based on methods developed in Ref.~\cite{burgarthOneBoundRule2022}); there, in particular, it was shown that the error committed when performing the RWA carries a nontrivial dependence, other than on the parameter $g=\lambda/\omega$, on the number $n$ of bosons of the \emph{specific} state whose evolution we are describing. In particular, the error committed when replacing $\mathrm{e}^{-\mathrm{i}tH_{\mrm{Rabi}}}$ with $\mathrm{e}^{-\mathrm{i}tH_{\mrm{JC}}}$ converges to zero as $\omega\to\infty$ for \emph{each} state; but, because of the dependence on the number of bosons $n$ (which might take arbitrarily large values), the supremum of the error on all possible states does not converge to zero. Mathematically, the RWA converges in the strong dynamical sense, but not in the uniform dynamical sense.

\begin{figure}[t]
    \centering
    \includeinkscape[pretex=\scriptsize, width=\textwidth]{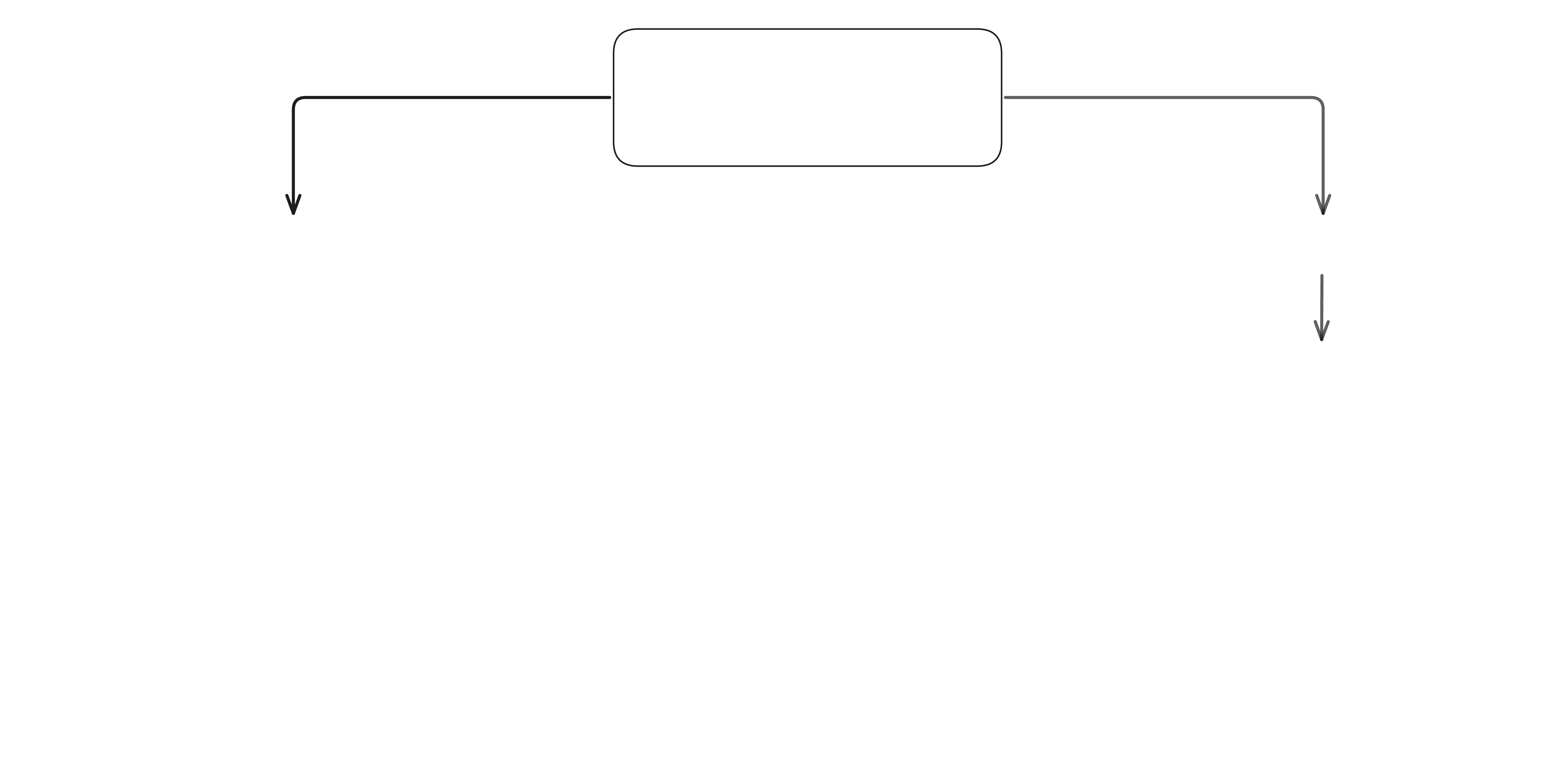}
    \caption{A schematic overview of some of many possible generalizations of the quantum Rabi model.}
    \label{fig:generalization-overview}
\end{figure}

In this paper we will push this analysis forward by studying the validity of the RWA for the multi-atom counterpart of the quantum Rabi model---the Dicke model:
\begin{equation}
    H_{\mrm{D}}=\frac{\omega_0}{2}\sum_{i=1}^N\sigma_z^{(i)}+\omega a^\dag a+\lambda\sum_{i=1}^N(\sigma^{(i)}_+ + \sigma^{(i)}_-)(a+a^\dag),
\end{equation}
which describes $N$ identical spins individually coupled with a monochromatic boson\linebreak field~\cite{dickeCoherenceSpontaneousRadiation1954}.
In the large $\omega$ regime, the dynamics induced by $H_{\mrm{D}}$ is expected to approach the one induced by the multi-atom counterpart of the Jaynes--Cummings model, which is usually denoted as the Tavis--Cummings model~\cite{Tavis_Cummings_1968}:
\begin{equation}
    H_{\mrm{TC}}=\frac{\omega_0}{2}\sum_{i=1}^N\sigma_z^{(i)}+\omega \ad \an+\lambda\sum_{i=1}^N(\sigma^{(i)}_+\an+\sigma^{(i)}_-\ad)
    .
\end{equation}
\Cref{fig:generalization-overview} shows a schematic overview of how the Dicke model relates to the Rabi model and a selection of other generalizations of it.
As we will show, the analysis of the RWA for the Dicke model will both offer mathematical and physical novelties, which justify the need for a separate discussion.
Let us briefly elaborate on both.

On the mathematical side, the multi-atom case exhibits additional (and nontrivial) intricacies with respect to the single-atom case. For one, the proof technique developed in Ref.~\cite{burgarthTamingRotatingWave2024} crucially relies on the existence of a subspace of the spin--boson Hilbert space which is invariant under the action of the evolution induced by $H_{\mrm{JC}}$ in the interaction picture. The choice performed in Ref.~\cite{burgarthTamingRotatingWave2024} relies on an explicit calculation of the latter dynamics, which, while possible for $N=1$, becomes quickly unpractical for $N>1$. The proof strategy developed in the present paper will be based on a \textit{different} choice of invariant domain which avoids the need for such a calculation and, as a result, is perfectly suited to the multi-atom case---as well as, in perspective, to further generalizations of the model (multimode fields, spin--spin interactions, etc.)

On the physical side, the presence of multiple spins interacting with the field offers fertile ground for (at least) two interesting, and strictly intertwined, questions of physical nature. First: the bounds developed in the paper will be (necessarily) state-dependent bounds. While in the single-atom case this essentially entails, as discussed, to understanding the dependence of the bounds on the number $n$ of photons in the state, in the presence of multiple atoms this also involves a rich dependence of the bounds on the spin part of the total spin--boson state. 

Secondly: while in the single-atom case the parameters dictating the validity of the RWA are two---the ratio $g=\lambda/\omega$, and the number $n$ of photons in the state---here a third actor comes into play: the number $N$ of atoms. How does the value of $N$ affect the precision of the RWA? 
Investigating this dependence numerically is unfeasible with regard to experimental setups in which, for example in solid state physics, \(10^{23}\) particles need to be considered~\cite{congDickeSuperradianceSolids2016a}. 
This question is particularly nontrivial if one considers the asymptotic regime $N\to\infty$ (thermodynamic limit), where a phase transition to a superradiant phase is known~\cite{Hepp_Lieb_1973}. Under which condition does the RWA still hold in this limit?

The scope of the present paper is precisely to analytically address these questions. We will find \emph{quantitative} bounds on the validity of the RWA for the Dicke model---thus, in the process, providing the first rigorous proof of validity of the approximation for the Dicke model with $N>1$ spins---and discuss the role of all parameters of the model in determining these bounds. 
We put particular focus on the dependence of the error introduced by the RWA with the parameters of the system and properties of the initial state.
Here, the total angular momentum \(S\) and the total number of excitations \(M\) in the initial state turn out to be the relevant quantities beyond the relation of the coupling strength \(\lambda\) to the frequency \(\omega\).
We find for states with finitely many photons that the difference between the Dicke dynamics and the Tavis--Cummings dynamics scales of order 
\begin{align}
    \frac{\lambda}{\omega} \abs*{\sin(\omega t)} \mathcal{O}\of*{S^{2} M ^{\nicefrac{1}{2}}} + 3\frac{\lambda^2}{\omega} \abs{t} \mathcal{O}\of*{S ^{\nicefrac{7}{2}} M }
\end{align}
by deriving an explicit error bound that can be easily evaluated.
A numerical analysis will also complement and validate our analytical results.

The manuscript is organized as follows. 
In \cref{ch:preliminaries} we define the Dicke and Tavis--Cummings models, as well as the invariant subspace we will employ in our proofs.
\Cref{ch:general-bound} builds the core of the present manuscript by stating the main theorem for general initial states, as well as showing its proof in detail. 
In \cref{sec:analysis}, we further analyse the spin-state dependence of the bound derived in \cref{ch:general-bound}, and compare the resulting error estimates with numerical results.
Final remarks and outlooks are presented in \cref{ch:conclusion}.
For readability, some proofs are omitted from the main text and included in the appendix.

\section{Preliminaries}\label{ch:preliminaries}

In this section, we provide basic definitions and preliminary results. \Cref{ch:finite-photon-subspace} introduces the subspace of states with finitely many photons, which will serve as the main setting for our analysis. The Dicke and Tavis–Cummings models are defined in \cref{ch:dicke-tc-hamiltonians}, along with the associated collective spin operators. Finally, \cref{ch:dicke-tc-interaction-picture} outlines the transition to the interaction picture.

\subsection{The finite photon subspace}\label{ch:finite-photon-subspace}
Let us consider the Hilbert space of square-integrable complex-valued functions \(\Ltwo\R\).
On the dense subspace of rapidly decreasing functions, the Schwartz space \(\mathcal{S}(\R)\), we define the creation 
and annihilation operators \(\ad,\an\) by
\cite{reedIIFourierAnalysis1975}
\begin{equation}
    \ad = \frac{1}{\sqrt 2}\of*{ x-\odv{}{x} },\qquad\an = \frac{1}{\sqrt 2}\of*{ x+\odv{}{x} }.
\end{equation}
These operators are known to satisfy the bosonic canonical commutation relations (\textit{CCR})
\begin{equation}\label{eq:CCR}
    \comm\an\ad = \an\ad - \ad\an = \one,
\end{equation}
and leave the Schwartz space invariant.

A complete orthonormal set for \(\Ltwo\R\) is given by the \emph{Hermite functions} \(\phi_n\), which are the Fock states in this representation of creation and annihilation operators.
Physically, in this representation, the \(n\)-th Hermite function \(\phi_n\) describes a boson state with \(n\) photons.
By construction, the creation, annihilation, and number operators \(\nn = \ad\an\) act on them as follows:
\begin{align}
	\label{eq:creation-on-number-state}
    \ad\phi_n &= \sqrt{n+1}\phi_{n+1}    
    \\
	\label{eq:annihilation-on-number-state}
    \an\phi_n &= \sqrt{n}\phi_{n-1}  
    \\
    \nn\phi_n &= n\phi_n. \label{eq:number-operator-def}
\end{align}
Besides, the Hermite functions form a complete orthonormal set of eigenvectors for the number operator \(\nn\), and hence \(\nn\) is essentially self-adjoint on \(\mathcal{S}(\R)\).
Also, we note that the Schwartz space can be written as
\begin{equation}\label{eq:schwartz-n}
    \mathcal{S}(\R) = \Set*{\sum_{n=0}^{\infty} c_n\phi_n \given c_n\in\C,  \forall \alpha \geq0 \colon \abs{c_n}^2 n^\alpha \xrightarrow[n\to\infty]{} 0}
    ,
\end{equation}
that is, the space of (infinite) linear combinations of Hermite polynomials with coefficients that decay with $n$ more quickly than any power law.

In the following, instead of working with the Schwartz space, we will consider the space of finite linear combinations of Hermite polynomials:
\begin{equation}\label{eq:finite-photons}
    \finitePhotonFock 
    = \Set*{\sum_{n=0}^{n_{\mathrm{max}}} c_n\phi_n \given c_n\in\C, n_{\mathrm{max}}\in\N}
    .
\end{equation}
This can be interpreted as the space of all states of the field containing a finite---albeit state-dependent, and thus with arbitrarily large cardinality---number of photon excitations. While being a proper subspace of $\mathcal{S}(\R)$, this is still dense in the Hilbert space $L^2(\R)$ by construction, since the Hermite functions form a complete orthonormal set. This choice will prove to be convenient for our purposes.

\subsection{The Dicke and Tavis--Cummings models}\label{ch:dicke-tc-hamiltonians}

We can now define the Dicke Hamiltonian \(H_{\mrm{D}}\) and the Tavis--Cummings Hamiltonian \(H_{\mrm{TC}}\) (at resonance) as the operators with domain \(\denseD = (\C^{2})^{\otimes N} \otimes \finitePhotonFock \subset \HH = (\C^{2})^{\otimes N} \otimes \Ltwo{\R}\), acting on it as follows:
\begin{align}
    \label{eq:def-dicke-hamiltonian}
    H_{\mrm{D}} &= \underbrace{\omega \SZ\otimes\one + \omega \one\otimes\nn}_{H_0} + \underbrace{\vphantom{\omega} \lambda\of{\s+ + \s-} \otimes (\an +\ad)}_{V_\mathrm{D}}     
    \\
    \label{eq:def-tavis-hamiltonian}
    H_{\mrm{TC}} &= \underbrace{\omega \SZ\otimes\one + \omega \one\otimes\nn}_{H_0} + \underbrace{\vphantom{\omega} \lambda\of{\s+\otimes\an + \s-\otimes\ad}}_{V_{\mathrm{TC}}},
\end{align}
where \(\omega\) is the frequency of the monochromatic field, \(\lambda\) is the strength of the coupling between the spins and the field, and
\begin{equation}
    S_{x} = \frac{1}{2} \sum_{i=1}^N \pauliI{x}{i}
    , \quad 
    S_{y} = \frac{1}{2} \sum_{i=1}^N \pauliI{y}{i}
    ,\quad 
    S_{z} = \frac{1}{2} \sum_{i=1}^N \pauliI{z}{i}
\end{equation}
denote collective spin operators, where \(\pauliI{w}{i} = \one ^ {\otimes (i-1)} \otimes \pauli{w} \otimes \one ^ {\otimes(N-i)}\) for \(w\in\{x,y,z\}\) are Pauli operators acting on the \(i\)-th spin.
They satisfy the usual angular momentum commutation relations \(\comm{S_{i}}{S_{j}} = \iu \tensor{\varepsilon}{_{k,l}^{m}}\s{}_m\).
Similarly,
\begin{equation}
    \s\pm 
    = \sum_{i=1}^N \pauliI{\pm}{i}
    = \Sx \pm \iu \Sy
\end{equation}
are collective spin-raising and lowering operators with 
\begin{equation}
    \pauliPlus = \begin{pmatrix}0&1\\0&0\end{pmatrix}, \quad \pauliMinus = \begin{pmatrix}0&0\\1&0\end{pmatrix}.
\end{equation}
Although the collective spin raising and lowering operators do not satisfy \emph{canonical anticommutation relations} for \(N>1\), as
\begin{equation}\label{eq:SplusSminusCommutator}
    \s+\s- + \s-\s+ = 2 \of*{\Sx^2+\Sy^2},
\end{equation}
they do satisfy, together with \(\SZ\), the commutation relations of the ladder algebra:
\begin{equation}\label{eq:ladder-algebra}
    \comm{\SZ}{\s\pm} = \pm\s\pm
    ,\quad 
    \commutator{\s+}{\s-} = 2\SZ.
\end{equation}

In both \cref{eq:def-dicke-hamiltonian,eq:def-tavis-hamiltonian}, the free Hamiltonian \(H_0\) describes the energy of \(N\in\N\) non-interacting spins and a monochromatic bosonic field with frequency \(\omega\); the models differ in the respective interaction terms \(V_D\) and \(V_{TC}\). Both Hamiltonians are essentially self-adjoint on \(\denseD\); with a slight abuse of notation, we will hereafter denote their unique self-adjoint extensions (their closures) by the same symbols \(H_{\mrm{D}}\) and \(H_{\mrm{TC}}\), respectively. 

We stress that the resonance assumption---that is, the fact that the spin and boson frequencies are assumed to be the same---is only taken for the sake of simplicity, and it could be relaxed at the price of more cumbersome calculations within the proof of \cref{thm:general-bound}; in any case,
the resonant case is the idealised situation in which the RWA is assumed to work best~\cite{burgarthTamingRotatingWave2024,Košata_Leuch_Kästli_Zilberberg_2022}.

One reason why it is useful to employ the RWA, i.~e. to consider \(H_{\mrm{TC}}\) instead of \(H_{\mrm{D}}\), is that the former is simpler to study analytically. 
For example, the Tavis--Cummings model is integrable: there is a complete orthogonal set of eigenstates for \(H_{\mrm{TC}}\) corresponding to the solution of a system of coupled multinomial (not differential) equations \cite{Bogoliubov_Bullough_Timonen_1996}.
As these eigenstates will be important in the following, we briefly introduce them.

To this purpose, let us begin by introducing \emph{Dicke states}, which occupy an important role in the study of the Dicke model, see e.g.~\cite{Qu_Zhang_Ni_Shan_David_Mølmer_2024,dickeCoherenceSpontaneousRadiation1954,Wu_Zhang_Wu_Su_Liu_Oxborrow_Shan_Mølmer_2024,Hepp_Lieb_1973,Tóth_2007,Gross_Haroche_1982}.
These states, denoted by \(\ket{S,m}\), are defined as the simultaneous eigenstates of the total collective spin operator \({\spin{}}^2 = \SX^2 + \SY^2+ \SZ^2\) and the collective spin operator in the \(z\)-direction, \(\SZ\), similar to the usual angular momentum eigenbasis.
The corresponding eigenvalues are
\begin{align}
    \label{eq:dicke-states-spin-number}
    {\spin{}}^2\ket{S,m} 
    &= S(S+1)\ket{S,m}
    \\
    \label{eq:dicke-states-magnetic-number}
    \SZ\ket{S,m}
    &= m\ket{S,m},
\end{align}
for \(S = N/2, N/2-1, \dots,1/2\text{ or }0\) (depending on whether $N$ is odd or even) and \(m=-S, -S+1, \dots, S-1, S\).
Both the Dicke and Tavis--Cummings Hamiltonians, \(H_D\) and \(H_{\text{TC}}\), commute with \({\spinalt{}}^2 = {\spin{}}^2 \otimes \one\), so that their eigenstates are characterised by a precise value of the total spin quantum number \(S\).
As we will see in \cref{ch:preliminaries}, \(H_{\mrm{TC}}\) additionally commutes with \(\M = \Szalt/2 + \nnalt\), which corresponds to the total number of excitations in the system.

Correspondingly, we denote by \(\tcpsi{S}{M}{\sigma}\) the joint eigenstates of \(H_{\mrm{TC}}\), the total collective angular momentum operator \({\spinalt{}}^2\), and \(\M\) with
\begin{align}
    H_{\mrm{TC}} \tcpsi{S}{M}{\sigma} 
    &= \tceig{S}{M}{\sigma} \tcpsi{S}{M}{\sigma}
    \\
    {\spinalt{}}^2 \tcpsi{S}{M}{\sigma} 
    &= S(S+1) \tcpsi{S}{M}{\sigma}
    \\
    \M \tcpsi{S}{M}{\sigma}
    &= (M - S) \tcpsi{S}{M}{\sigma}
    ,
\end{align}
where \(M \in \N\) is an integer, \(S\) given as in \cref{eq:dicke-states-spin-number} for the Dicke states, and \(\sigma = 1,2,\dots, K_{S,M}\), with \(K_{S,M} = \min\Set{2S,\,M}+1\), is an additional index.
An explicit expression of these states in terms of Dicke states was found by Bogoliubov et al. using the \emph{quantum inverse method}, which is related to the \emph{Bethe ansatz} \cite{korepinQuantumInverseScattering1993}.
They are given by 
\begin{equation}
    \label{eq:TC-eigenstates}
    \tcpsi{S}{M}{\sigma} 
    = \prod_{i=1}^{M} \of*{\tcpar{i}{\sigma}\adalt - \salt+} \ket{S,-S}\otimes\phi_0
\end{equation}
with \(\salt+ = \spin+\otimes\one\) and \(\adalt = \one\otimes\ad\), and where, in the fully resonant case, 
\(\tcparvec{\sigma} = (\tcpar{1}{\sigma}, \dots, \tcpar{M}{\sigma})\) is one of the \(K_{S,M}\) solutions of the Bethe equations
\begin{align}
    \label{eq:TC-bethe-equations}
    \frac{2S}{\tcpar{n}{\sigma}} - \tcpar{n}{\sigma} - \sum_{\substack{j=1\\j\neq n}}^{M} \frac{2}{\tcpar{n}{\sigma} - \tcpar{j}{\sigma}} 
    = 0,
    \quad
    \forall n = 1, \dots, M
\end{align}
and \(\sigma\) labels different solutions for the same values of \(S\) and \(M\). The corresponding eigenvalues are given by
\begin{equation}
    \label{eq:TC-eigenenergies}
    \tceig{S}{M}{\sigma} 
    = \omega(M-S) - \lambda \sum_{i=1}^{M} \tcpar{i}{\sigma}
    .
\end{equation}
Importantly, each of these eigenstates is a member of the finite photon subspace: \(\tcpsi{S}{M}{\sigma} \in \denseD\) for every value of \(S,M,\sigma\). This readily follows from the fact that each of these states results from applying \(M<\infty\) creation operators \(\adalt\) to \(\phi_0\). 

We note that, although closed-form solutions of the Bethe equations~\eqref{eq:TC-bethe-equations} are not known for $M>2$, they are still far more amenable to numerical treatment than computing the full evolution operator generated by the Tavis–Cummings Hamiltonian. As such, while some of our error bounds require solving the Bethe equations numerically, this remains significantly more efficient than simulating the full quantum dynamics.

\subsection{The Dicke and Tavis--Cummings Hamiltonians in the interaction picture}\label{ch:dicke-tc-interaction-picture}

We need to study our problem in the rotating frame (interaction picture) with respect to the free Hamiltonian \(H_0\). In this representation, the two operators \(H_{\mrm{D}}\) and \(H_{\mrm{TC}}\) become
\begin{align}
    H_1(t) &= \mrm{e} ^ {\iu t H_0} \of*{H_{\mrm{D}}-H_0} \mrm{e} ^ {-\iu t H_0} \label{eq:def-H1}\\
	H_2(t) &= \mrm{e} ^ {\iu t H_0} \of*{H_{\mrm{TC}}-H_0} \mrm{e} ^ {-\iu t H_0} \label{eq:def-H2}
    .
\end{align}

Recall that a vector \(\psi\in\HH\) is said to be \textit{analytic} for an operator \(A\colon \dom(A) \subseteq \HH \to \HH\) if the two following conditions hold: \(\psi\in\dom(A^k)\) for any \(k\in\N\), and
\begin{equation}\label{eq:analytic-vector-condition}
    \sum_{k=0}^\infty \frac{t^k}{k!}\norm*{A^k \psi} < \infty,
\end{equation}
for some \(t\in\R\)~\cite{nelsonAnalyticVectors1959}.
Importantly, when acting on an analytic vector \(\psi\) for self-adjoint \(A\), the unitary dynamics \(\mrm{e} ^ {-\iu t A} \psi\) is simply given by a power series for any \(t\in\R\) such that \cref{eq:analytic-vector-condition} holds~\cite[prop. 9.25]{morettiSpectralTheoryQuantum2017}, that is,
\begin{equation}\label{eq:power-series}
    \mrm{e} ^ {-\iu t A} \psi = \sum_{k=0}^\infty \frac{(- \iu t)^k}{k!} A^k \psi
    .
\end{equation}
Coming back to our problem, every vector in \(\denseD\) is analytic for all relevant operators in the following calculations. Specifically: the analyticity with respect to \(H_{\mrm{D}}\) and \(H_{\mrm{TC}}\) can be shown directly by standard methods~\cite[Sec. X.6]{reedIIFourierAnalysis1975};
the analyticity of \(\denseD\) for \(H_0\) follows trivially from the fact that vectors in \(\denseD\) are eigenvectors of \(H_0\).
For the same reason, \cref{eq:analytic-vector-condition} holds for any \(t\in\R\) in case of \(A\) being \(H_0\).

For \(H_1(t)\) acting on \(\psi\in\denseD\), we explicitly have 
\begin{equation}\label{eq:H1}
    H_1(t)\psi = \lambda\of*{\spin+\otimes \an + \spin-\otimes\ad + \mrm{e} ^ {\iu 2\omega t}\spin+\otimes \ad + \mrm{e} ^ {-i2\omega t}\spin-\otimes \an}\psi\\
\end{equation}
and, as \(\mrm{e} ^ {\iu H_0 t}\) acts on \(\psi\) as a power series, we get
\begin{equation}\label{eq:H2}
    H_2(t) \psi = H_2 \psi = \lambda\of{\s+\otimes\an + \s-\otimes\ad} \psi,
\end{equation}
the latter, in particular, being time-independent (see also \cref{thm:TC-H0-commutation}).

We also define the operators
\begin{align}
    \label{eq:def-U1}
    U_1(t) &= \mrm{e} ^ {\iu H_0 t} \mrm{e} ^ {-\iu t H_{\mrm{D}}}
    \\
    \label{eq:def-U2}
    U_2(t) &= \mrm{e} ^ {\iu H_0 t} \mrm{e} ^ {-\iu t H_{\mrm{TC}}}
\end{align}
that, as will be shown later in \cref{thm:strong-derivatives}, describe the unitary dynamics generated by \(H_1(t)\) and \(H_2\), respectively.
To compare \(\mathrm{e}^{-\iu t H_{\mathrm{D}}}\) and \(\mathrm{e}^{-\iu t H_{\mathrm{TC}}}\), we focus on their interaction-picture counterparts \(U_1(t)\) and \(U_2(t)\). Since the transformation to the interaction picture is unitary and common to both evolutions, the difference between \(U_1(t)\) and \(U_2(t)\) coincides with that between the original Schrödinger-picture evolutions, allowing us to carry out the comparison in the interaction picture.

\section{Error bound for the RWA of the Dicke model}\label{ch:general-bound}

We can now proceed to state and prove the main result of the present manuscript.
\begin{theorem}\label{thm:general-bound}
	Let $\mrm{e} ^ {-\iu t H_{\mrm{D}}}$, $\mrm{e} ^ {-\iu t H_{\mrm{TC}}}$ be the evolution groups generated by the Tavis--Cummings and Dicke Hamiltonians, \(H_{\mrm{D}}\) and \(H_{\mrm{TC}}\). 
    For every eigenstate \(\tcpsi{S}{M}{\sigma}\) of the Tavis--Cummings Hamiltonian given as in \cref{eq:TC-eigenstates}, the following estimate holds for any $t>0$: 
	\begin{gather}
		\label{eq:general-bound}
		\norm[\big]{ \of[\big]{\mrm{e} ^ {-\iu t H_{\mrm{TC}}} - \mrm{e} ^ {-\iu t H_{\mrm{D}}}}\tcpsi{S}{M}{\sigma} } 
        \leq \frac{\lambda}{\omega}\abs*{\sin(\omega t)} \min_\pm f^{\pm}_C(\tcpsi{S}{M}{\sigma}) 
        + 3\abs{t}\frac{\lambda^2}{\omega} f_L(\tcpsi{S}{M}{\sigma})
	\shortintertext{with \(f_{C},f_{L}\colon\:\denseD\to\R\) given by}
    \label{eq:def-fC-fL}
	\begin{aligned}
     &f^{\pm}_C(\psi) = \eof*{
        2\norm[\big]{\of{ {\spinalt{}}^2 - \Sz^2} \psi} 
        + 2 \norm[\big]{\of{\salt\pm}^2\psi}
            }^{\nicefrac{1}{2}}
        \norm[\big]{\of{\nnalt + 2 }\psi}^{\nicefrac{1}{2}}
	   \\
		&\!\begin{multlined}
			f_L(\psi) = 
			\biggl[	
				\norm[\big]{\of{\spinalt+ \spinalt-}^2\psi}^{\nicefrac{1}{2}} 
				+ \norm[\big]{\of{\spinalt- \spinalt+}^2\psi}^{\nicefrac{1}{2}}
				+ {\norm[\big]{\of{\spinalt+}^2\of{\spinalt-}^2 \psi}}^{\nicefrac{1}{2}}  
				\\
				+ {\norm[\big]{\of{\spinalt-}^2\of{\spinalt+}^2 \psi}}^{\nicefrac{1}{2}}
			\biggr]
			\cdot\norm[\big]{\of{\nnalt +4}^2 \psi}^{\nicefrac{1}{2}}
		\end{multlined}
		\end{aligned}
	\end{gather}
	and where \(\spinalt\pm = \spin\pm\otimes\one\) and \(\nnalt = \one\otimes\nn\).
	Consequently, for every \(\psi\in\HH\) the RWA for the Dicke Hamiltonian converges in the limit $\omega\to\infty$:
    \begin{equation}\label{eq:dicke-tc-convergence}
        \norm[\big]{ \of[\big]{\mrm{e} ^ {-\iu t H_{\mrm{TC}}} - \mrm{e} ^ {-\iu t H_{\mrm{D}}}}\psi } \xrightarrow[\omega \to \infty]{} 0
        ,
	\end{equation}
  the convergence being uniform in $t$ for every compact set of times.
 \end{theorem}

This provides a quantitative estimate on the error on the evolution of $\tcpsi{S}{M}{\sigma}$ committed when replacing the Dicke Hamiltonian with the Tavis--Cummings one---that is, the error of the RWA.
We remark that, while the upper bound \cref{eq:general-bound} gives rise to a quantitative estimate for eigenstates of the Tavis--Cummings Hamiltonian, the fact that the RWA converges as $\omega\to\infty$ (\cref{eq:dicke-tc-convergence}) holds for \emph{all} states of the system. 
As we will see, this is ultimately a consequence of the fact that the eigenstates form an orthonormal basis.

The main ingredient of the proof of \cref{thm:general-bound} will be an integration-by-parts method (\cref{ch:partial-integration-method}) previously developed for bounded operators~\cite{burgarthOneBoundRule2022}, and subsequently applied to the Rabi and Jaynes--Cummings Hamiltonians~\cite{burgarthTamingRotatingWave2024}. 
The remainder of this section shall be devoted to this proof, which is split into multiple lemmas and subsections.
In \cref{ch:denseD-invariance}, preliminary results necessary for applying the integration-by-parts method are stated and proven;
the main method itself is revisited in \cref{ch:partial-integration-method};
\cref{ch:calculating-general-bound} concludes the proof of \cref{thm:general-bound} with the explicit calculation of the bound;
finally, \cref{ch:worst-case-bound} presents an alternative to \cref{thm:general-bound} with a similar proof.

\subsection{Preliminary results}\label{ch:denseD-invariance}

As detailed in~\cite{burgarthTamingRotatingWave2024}, in order to evaluate the difference between two dynamics $U_1(t)$, $U_2(t)$ on some ``nice'' space of vectors without incurring in domain issues, it is useful to prove the invariance of said space under one of the two dynamics.
Here, we show invariance of \(\denseD\) under \(U_2(t)\), by a symmetry argument.

To this end, we will begin by pinpointing a crucial difference between the two Hamiltonians, ultimately caused by the lack of the \emph{counter-rotating} terms \(\mrm{e} ^ {\iu 2\omega t}\spin+\otimes \ad\) and \(\mrm{e} ^ {-i2\omega t}\spin-\otimes \an\) of \(H_1(t)\) in \(H_2\): the evolution induced by \(H_2\) preserves the energy of the \emph{non-interacting} system, that is:

\begin{lemma}\label{thm:TC-H0-commutation}
    The dynamics generated by the Tavis--Cummings Hamiltonian \(H_{\mrm{TC}}\)~\eqref{eq:def-tavis-hamiltonian} commutes with the free Hamiltonian \(H_0\):
    \begin{equation}\label{eq:TC-H0-commutation}
        \comm{\mrm{e} ^ {-\iu t H_{\mrm{TC}}}}{H_0}\psi = 0, \quad \forall t\in \R
    \end{equation}
    for any \(\psi\in\denseD\).
\end{lemma}
\begin{proof}
    Let \(H_{\mrm{TC}}\) and \(H_0\) be given as in \cref{eq:def-dicke-hamiltonian}. 
    Let \(\psi\in \denseD\). 
    As \(\denseD\) consists of finite linear combinations of eigenvectors of \(H_0\), \(\psi\) is an analytic vector for \(H_0\): as such, the condition~\eqref{eq:analytic-vector-condition} holds for any \(t\in\R\).
    Therefore, the action of the exponential \(\mrm{e} ^ {\iu t H_0}\) on $\psi$ is directly given by the power series~\cite[Prop. 9.25]{morettiSpectralTheoryQuantum2017}
    \begin{equation}\label{eq:H-0-power-series}
        \mrm{e} ^ {\iu t H_0}\psi = \sum_{k=0}^\infty \frac{(it)^k}{k!} H_0^k \psi, \quad \forall \psi\in\dom(A).
    \end{equation}
    Furthermore, a direct calculation using the commutation relations \eqref{eq:CCR}--\eqref{eq:ladder-algebra} shows that the two Hamiltonians themselves commute on \(\psi \in \denseD\):
    \begin{equation}
        \comm{H_{\mrm{TC}}}{H_0}\psi = 0.
    \end{equation}
    Besides, \(H_0\denseD\subseteq\denseD\) and \(H_{\mrm{D}}\denseD\subseteq\denseD\).
    By these properties, we have
    \begin{align}
		H_{\mrm{TC}}\mrm{e} ^ {-\iu s H_0}\psi 
        &= \sum_{k=0}^\infty \frac{(-\iu s)^k}{k!} H_{\mrm{TC}} H_0^k\psi \\
		&= \sum_{k=0}^\infty \frac{s^k}{k!} H_0^k H_{\mrm{TC}} \psi \\
		&= \mrm{e} ^ {-\iu sH_0} H_{\mrm{TC}} \psi,
	\end{align}
    i.~e. \(H_0\) is a unitary symmetry of the Tavis--Cummings model. 
    This is equivalent to \(\comm{\mrm{e} ^ {-\iu t H_{\mrm{TC}}}}{U_0(s)} = 0\), which in turn is equivalent to \cref{eq:TC-H0-commutation} due to~\cite[Thm. 9.41]{morettiSpectralTheoryQuantum2017}.
\end{proof}

The symmetry proven in \cref{thm:TC-H0-commutation} will now allow us to show invariance of \(\denseD\) under \(\mrm{e} ^ {-\iu t H_{\mrm{TC}}}\):
\begin{lemma}\label{thm:denseD-invariance-under-UTC}
    The subspace \(\denseD\subseteq \HH\) of states with finitely many photons is invariant under the dynamics \(\mrm{e} ^ {-\iu t H_{\mrm{TC}}}\):
    \begin{equation}\label{eq:denseD-invariance-TC}
        \mrm{e} ^ {-\iu t H_{\mrm{TC}}}\denseD \subseteq \denseD,
    \end{equation}
    for any \(t\in\R\).
\end{lemma}
\begin{proof}
     Consider an eigenvector \(\psi_0 = \ket{\underline s}\otimes\phi_n \) of \(H_0\) with eigenvalue \(\lambda_{\underline s, n}\), where \(\ket{\underline s}\) is an eigenvector of \(\SZ\) and \(\phi_n\) is the $n$th Hermite function, i.~e. an eigenvector of \(\nn\).
    By definition, \(\psi_0\in\denseD\).
    By \cref{thm:TC-H0-commutation} we then have
	\begin{equation}
		H_0 \mrm{e} ^ {-\iu t H_{\mrm{TC}}} \psi_0 = \mrm{e} ^ {-\iu t H_{\mrm{TC}}} H_0 \psi_0 = \lambda_{\underline s, n} \mrm{e} ^ {-\iu t H_{\mrm{TC}}}\psi_0
	\end{equation}
	and hence \(\mrm{e} ^ {-\iu t H_{\mrm{TC}}}\psi_0\) is an eigenvector of \(H_0\).
	On the other hand, every eigenvector of \(H_0\) is in \(\denseD\) because \(\Sz\) is bounded which therefore concludes the proof.
\end{proof}
\begin{remark}
We can compare this lemma with~\cite[Lemma 1.1]{burgarthTamingRotatingWave2024}, where the invariance of $\C^2\otimes\mathcal{S}(\R)$ under the evolution generated by the Jaynes--Cummings Hamiltonian is proven. In such a case, an explicit computation of the dynamics, which readily becomes cumbersome for $N>1$ spins, was needed. Here, our choice of a smaller (but still dense) domain enabled us to prove invariance with a simpler argument which did not require such a calculation. As such, we expect this choice to be also suitable for more complicated models involving the RWA, cf. the discussion in \cref{ch:conclusion}. 
\end{remark}

The properties shown in \cref{thm:TC-H0-commutation} and \cref{thm:denseD-invariance-under-UTC} readily translate in the interaction picture---that is, from $\mathrm{e}^{-\iu tH_{\mrm{TC}}}$ to $U_2(t)$:

\begin{corollary}\label{thm:denseD-invariance}
    The subspace \(\denseD\subseteq \HH\) of states with finitely many photons is invariant under the dynamics \(U_2(t)\) as defined in \cref{eq:def-U2}:
    \begin{equation}\label{eq:denseD-invariance}
        U_2(t)\denseD \subseteq \denseD,
    \end{equation}
    for any \(t\in\R\).
\end{corollary}
\begin{proof}
    As \(\denseD\) is spanned by eigenvectors for \(H_0\), \(\mrm{e} ^ {-\iu t H_0}\) leaves \(\denseD\) invariant.
    Due to \cref{thm:denseD-invariance-under-UTC}, \(\mrm{e} ^{-\iu t H_{\mrm{TC}}}\) does so as well.
    Hence, \(U_2(t) = \mrm{e} ^{\iu t H_0} \mrm{e} ^ {-\iu t H_{\mrm{TC}}}\) leaves \(\denseD\) invariant.
\end{proof}

\begin{corollary}\label{thm:U2-H0-commutation}
    The dynamics  \(U_{2}(t)\) generated by the Tavis--Cummings Hamiltonian in the rotating frame with respect to \(H_0\) commutes with the free Hamiltonian \(H_0\):
    \begin{equation}\label{eq:U2-H0-commutation}
        \comm{U_{2}(t)}{H_0}\psi = 0, \quad \forall t\in \R
    \end{equation}
    for any \(\psi\in\denseD\)
    .
\end{corollary}
\begin{proof}
It follows directly from \cref{eq:def-U2} and \cref{thm:TC-H0-commutation}.
\end{proof}

Furthermore, \cref{thm:denseD-invariance-under-UTC} allows us to show the relation between \(H_2\) and \(U_2(t)\) anticipated in \cref{ch:dicke-tc-interaction-picture}: \(-\iu H_2\) is the strong derivative of \(U_2(t)\):
\begin{lemma}\label{thm:strong-derivatives}
    For the dynamics in the interaction picture, \(U_1(t)\) and \(U_2(t)\), as defined per \cref{eq:def-U1,eq:def-U2}, the following properties hold for all \(\psi\in\denseD\):
    \begin{align}
    	\odv{}{t} U_1(t)\adj\psi &= \iu  U_1(t)\adj H_1(t) \psi 
        \\
        \label{eq:strong-derivative-U2}
    	\odv{}{t} U_2(t)\psi &= -\iu H_2(t) U_2(t)\psi
    \end{align}
    where \(H_1(t)\) and \(H_2(t)\) are given by \cref{eq:def-H1,eq:def-H2}.
\end{lemma}
\begin{proof}
	Let \(\psi \in\denseD\).
	Because of \cref{thm:denseD-invariance-under-UTC}, \(\mrm{e} ^ {-\iu t H_{\mrm{TC}}}\psi\in\denseD\subseteq\dom(H_0)\).
	Therefore, we can use the product rule for the strong derivative, thus obtaining
	\begin{align}
		\odv{}{t} U_2(t)\psi
		&= \odv{}{t} \of{\mrm{e} ^ {\iu tH_0} \mrm{e} ^ {- \iu t H_{\mrm{TC}}}\psi}\\
		&= \odv{}{t}\of{\mrm{e} ^ {\iu t H_0}} \mrm{e} ^ {-\iu t H_{\mrm{TC}}}\psi + \mrm{e} ^ {\iu tH_0}\odv{}{t}\of{ \mrm{e} ^ {-\iu t H_{\mrm{TC}}}}\psi\\
		&= -\iu \mrm{e} ^ {\iu t H_0} (H_{\mrm{TC}}-H_0) \mrm{e} ^ {-\iu t H_{\mrm{TC}}}\psi\\
		&= -\iu \mrm{e} ^ {\iu t H_0} (H_{\mrm{TC}}-H_0) \mrm{e} ^ {-\iu t H_0} \mrm{e} ^ {\iu t H_0} \mrm{e} ^ {-\iu t H_{\mrm{TC}}}\psi\\
		&= -\iu H_2(t) U_2(t)\psi.
	\end{align} 
    As \(\denseD\) is spanned by eigenvectors of \(H_0\), \(\mrm{e} ^ {-\iu t H_0}\) leaves \(\denseD\) invariant and \(\mrm{e} ^ {-\iu t H_0}\psi\in\denseD\subseteq\dom(H_D)\).
    Therefore, a similar calculation proves the claim.
\end{proof}

As \(H_2(t) = H_2\) is time-independent and essentially self-adjoint, \cref{eq:strong-derivative-U2} implies that \(U_2(t)\) is generated by the self-adjoint closure of \(H_2\). With a slight abuse of notation, we will thus write \(U_2(t) = \mrm{e} ^ {-\iu tH_2}\).

\subsection{Integration-by-parts method}\label{ch:partial-integration-method}

The invariance proven in \cref{thm:denseD-invariance} will be crucial for the application of the integration-by-parts method.
The idea of this method is to obtain the difference \(U_2(t)-U_1(t)\) as the boundary term occurring from some integration by parts, in such a way to express it in terms of quantities that vanish in the limit $\omega\to\infty$. As the difference between the Hamiltonians $H_1(t)$ and $H_2$ does \textit{not} converge strongly in this limit, the idea is to apply integration by parts \emph{twice}. Let us elaborate on that.

Let \(\psi\in\denseD\).
Due to \cref{thm:denseD-invariance}, we have \(U_2(t)\psi\in\denseD\).
Then, by using \cref{thm:strong-derivatives,thm:denseD-invariance}, an integration by parts yields
\begin{align}
    \iu\of{U_2(t)-U_1(t)}\psi 
    &=	\iu\of{U_1\adj(t) U_2(t)-\one}\psi\\
    &= \iu U_1(t)\bigl[U_1\adj(s)U_2(s)\psi \bigr]^{s=t}_{s=0} \\
    &= \iu U_1(t) \int_0^t \odv{}{s} U_1\adj(s)U_2(s) \psi\dd s \\
    &=  U_1(t) \int_0^t U_1\adj(s) \bigl( H_2-H_1(s) \bigr) U_2(s) \psi\dd s. \label{eq:partial-integration-part1}
\end{align}
Taking norms and applying the triangle inequality at the right-hand side, we obtain bound
\begin{equation}
    \norm{\of{U_2(t)-U_1(t)}\psi } \leq \int_0^t \norm{\of{H_2-H_1(s)} U_2(s) \psi}\dd s
\end{equation}
which essentially depends on the norm of the difference of the Hamiltonians applied to $\psi$. However, the latter does not converge to \(0\) in the case under consideration (and it does not do so in many cases), as the difference between the two Hamiltonians does not vanish as $\omega\to\infty$.

The integration-by-parts method circumvents this problem by performing a second integration by parts, which introduces the \emph{integrated action} of the difference between the Hamiltonians
again defined on $\denseD$: that is, for all $\psi\in\denseD$,
\begin{equation}\label{eq:integrated-action-def}
	S_{21}(t)\psi \coloneqq \int_0^t \of{H_2 - H_1(s)}\psi \dd s,
\end{equation}
the integral being understood in the Bochner sense~\citetext{\citealp{bochnerIntegrationFunktionenDeren1933}; \citealp[p.~115--117,~119]{reedMethodsModernMathematical1972}; \citealp[ch.~V.5]{yosidaFunctionalAnalysis1995}}.

The next lemma gives an explicit expression for \(S_{21}(t)\) and summarises some of its relevant properties.
\begin{lemma}\label{thm:integrated-action}
	The integrated action \(S_{21}(t)\) defined by \cref{eq:integrated-action-def} takes the form
	\begin{equation}\label{eq:integrated-action}
		S_{21}(t)\psi = -\frac{\lambda}{\omega}\sin\of{\omega t}\of*{\mrm{e} ^ {\iu\omega t} \spin+\otimes\ad + \mrm{e} ^ {-i\omega t}\spin-\otimes\an} \psi,
	\end{equation} 
	for any \(\psi\in\denseD\).
	Furthermore, 
	\begin{gather}
		\label{eq:integrated-action-invariance} 
		S_{21}(t)\denseD \subseteq \denseD			
	\shortintertext{and}
		\label{eq:integrated-action-derivative}
		\odv{}{t} S_{21}(t)\psi = (H_2-H_1(t))\psi 	
	\end{gather}
 	for any \(\psi\in\denseD\) and \(t\in\R\).
\end{lemma}
\begin{proof}
	Let \(\psi\in\denseD\).
	By \cref{eq:H1,eq:H2}, \(\of{H_2 - H_1(s)}\psi\) is integrable, and
	\begin{align}
		S_{21}(t) \psi
		&= -\lambda\int_0^t \of*{ \mrm{e} ^ {\iu 2\omega s}\spin+\otimes\ad + \mrm{e} ^ {-\iu 2\omega s}\spin-\otimes\an} \psi\dd s\\
		&= -\frac{\lambda}{\iu 2\omega}\of*{ \of{\mrm{e} ^ {\iu 2\omega t}-1}\spin+\otimes\ad - \of{\mrm{e} ^ {-\iu 2\omega t}-1}\spin-\otimes\an }\psi\\
		&= -\frac{\lambda}{\iu 2\omega}
		\underbrace{\of*{\mrm{e} ^ {\iu\omega t} - \mrm{e} ^ {-\iu\omega t}}}_{= 2 \iu\sin\of{\omega t}}
		\of*{ \mrm{e} ^ {\iu\omega t}\spin+\otimes\ad + \mrm{e} ^ {-\iu\omega t}\spin-\otimes\an }\psi
        ,
	\end{align}
	which proves \cref{eq:integrated-action}.
	The invariance~\eqref{eq:integrated-action-invariance} holds because \(\spin+\otimes\ad\) and \(\spin-\otimes\an\) leave \(\denseD\) invariant and  \(\denseD\) is a vector space.
	Finally, \cref{eq:integrated-action-derivative} follows from the fundamental theorem of calculus.
\end{proof}

\begin{lemma}[integration-by-parts method]\label{thm:partial-integration}
	For any \(\psi\in\denseD\), we have
	\begin{gather} \label{eq:partial-integration}
	\begin{split}
		\iu\of{U_2(t)-U_1(t)} \psi &= S_{21}(t)U_2(t)\psi\\
		&\quad+ \iu \int_0^t U_1(t)U_1\adj(s)
		\bigl( S_{21}(s)H_2 - H_1(s)S_{21}(s) \bigr) U_2(s)\psi \dd s.
	\end{split}
	\raisetag{3.5em}
	\end{gather}
\end{lemma}
\begin{proof}
	Let \(\psi\in\denseD\).
	Due to \cref{thm:integrated-action}, in addition to \(U_2(t)\psi\in\denseD\), we also have \(S_{21}(t)U_2(t)\psi\in\denseD\) and \(S_{21}(t) H_2 U_2(t)\psi\in\denseD\).
	By noticing that the derivative
	\begin{align}
		\begin{split}
			\odv{}{s} \bigl( U_1\adj(s)S_{21}(s)U_2(s) \bigr) \psi&= \iu U_1\adj(s) \of*{H_1(s)S_{21}(s) - S_{21}(s)H_2}U_2(s)\psi \\
			&\quad + U_1\adj(s)\of*{H_2-H_1(s)}U_2(s)\psi
		\end{split}
	\end{align}
	contains the integrand of \cref{eq:partial-integration-part1}, we can perform another integration by parts:
	\begin{align}
	\begin{split}
		\iu\of{U_2(t)-U_1(t)}\psi 
		&= U_1(t)\int_0^t \bigg[ \odv{}{s}\of*{U_1\adj(s)S_{21}(s)U_2(s)\psi}
        \\
		&\quad + \iu U_1\adj(s) \of*{ S_{21}(s)H_2 -H_1(s)S_{21}(s) }U_2(s)\psi \bigg]\dd s 
	\end{split}\\
	\begin{split}
	&= S_{21}(s)U_2(s)\psi 
    \\
	&\quad +\iu \int_0^t U_1(t)U_1\adj(s) \of{ S_{21}(s)H_2 - H_1(s)S_{21}(s) }U_2(s)\psi,
	\end{split}
	\raisetag{3.5em}
	\end{align}
	which concludes the proof.
\end{proof}
This property can be used, as we will see, to bound the difference between the two dynamics $U_2(t)$ and $U_1(t)$ in terms of the difference of their \textit{actions}, rather than the Hamiltonians themselves. This is indeed useful in those cases in which the difference between two Hamiltonians is only small when \emph{averaged} in time. As we will see in \cref{ch:calculating-general-bound}, this indeed happens in our case.

\subsection{Calculating the bound}\label{ch:calculating-general-bound}
We will now proceed towards the proof of \cref{thm:general-bound} by utilizing \cref{thm:partial-integration}. Some further intermediate results will be needed.

To begin with, let us recall a useful inequality---hereafter referred to as the \emph{Cauchy--Schwarz trick}---that will be repeatedly employed in our computations to bound the norm of quantities involving tensor product of spin and bosonic operators in terms of their components separately.
Consider the tensor product \(A\otimes B\) of two densely defined operators \(A\colon\dom(A)\subseteq\HH_1\to\HH_1\) and \(B\colon\dom(B)\subseteq\HH_2\to\HH_2\) on Hilbert spaces \(\HH_1\) and \(\HH_2\). 
Then, as a direct consequence of the Cauchy--Schwarz inequality, we have
\begin{equation}\label{eq:Cauchy--Schwarz-trick}
	\norm{(A\otimes B)\psi} \leq \sqrt{\norm{A\adj A\otimes \one \psi} \norm{\one\otimes B\adj B \psi}}.
\end{equation}

We present now a general bound for arbitrary polynomials of creation and annihilation operators:
\begin{lemma}\label{thm:bosonic-polynomial-bound}
	Let $k,l\in\N$, and $p^{(k)}_l(a,a^\dag):\finitePhotonFock\subset\Ltwo{\R}\rightarrow\Ltwo{\R}$ defined by
	\begin{equation}
		p^{(k)}_l(\an,\ad) = 
		\lambda_1 m^{(k_1)}_1 (\an,\ad) 
		+ \dots 	
		+ \lambda_l m^{(k_l)}_l(\an,\ad),
	\end{equation}
	with \(\lambda_i\in \C\), and each of the \(l \in \N\) terms above being a monomial in $a$ and $a^\dag$ of order \(k_i\leq k\in\N\):
    \begin{equation}\label{eq:def-monomial}
        m^{(k_i)}_i(\an,\ad) = a^{\#_1} \dots a^{\#_{k_i}}
        ,
    \end{equation}
    where each \(a^{\#}\) stands for either \( \ad \) or \( \an \).
    Then, for any \(\phi \in \finiteFock\), the following estimate holds:
	\begin{equation} \label{eq:bosonic-polynomial-bound}
		\norm*{p^{(k)}_l(\an,\ad)\phi} 
		\leq l 
		\,\max_{1\leq i\leq l}\cof{\abs{\lambda_i}}
		\, \norm*{\sqrt{(\nn+1)\dots(\nn+k)}\phi}.
	\end{equation}
\end{lemma}
\begin{proof}
	We first prove the bound for a monomial \(m^{(k)}(\an, \ad) = a^{\#_1}\dots a^{\#_k}\) of order \(k\) in the bosonic creation and annihilation operators and for \(\phi\in \finiteFock\). 
	For every \(\phi_n\), using \cref{eq:creation-on-number-state,eq:annihilation-on-number-state} we get 
	\begin{equation}
		\norm*{m^{(k)}(\an, \ad)\phi_n} \leq \norm*{\sqrt{(n+1)\dots(n+k)} \phi_n} 
	\end{equation}
	by the monotonicity of \(\sqrt{n}\). Then, for an arbitrary finite-photon state \(\phi = \sum_{n=0}^{N} c_n \phi_n \in \finiteFock \), \(c_n\in\C\), we have 
	\begin{align}
		\norm*{m^{(k)}(\an,\ad)\phi} 
		&= \sum_{n=0}^{N} \abs{c_n} \norm{m^{(k)}(\an,\ad)\phi_n}  
		\\
		&\leq \sum_{n=0}^{N} \abs{c_n} \sqrt{(n+1)\dots(n+k)} \norm*{\phi_n } 
		\\
		\label{eq:bosonic-monomial-bound}
		&= \norm*{\sqrt{(\nn+1)\dots(\nn+k)} \phi},
	\end{align}
	where, in the first and last line, we used the orthogonality of the \(\phi_n\).

	From this bound for monomials, \cref{eq:bosonic-polynomial-bound} directly follows by the triangle inequality, thus concluding the proof.
\end{proof}

We now state and prove a bound for the difference between the evolutions of arbitrary states in the finite-photon subspace.

\begin{prop}\label{thm:intermediate-bound}
    Let $\mrm{e} ^ {-\iu t H_{\mrm{D}}}$, $\mrm{e} ^ {-\iu t H_{\mrm{TC}}}$ be the evolution groups generated by the Tavis--Cummings and Dicke Hamiltonians, \(H_{\mrm{D}}\) and \(H_{\mrm{TC}}\), and let \(U_2(t) = \mrm{e}^{-\iu t H_2}\) be the evolution group generated by the Tavis--Cummings Hamiltonian in the interaction picture as in \cref{eq:H2}. 
    For every \(\psi\in\denseD\), the following estimate holds for any $t>0$: 
	\begin{gather}
		\label{eq:intermediate-bound}
		\norm[\big]{ \of[\big]{\mrm{e} ^ {-\iu t H_{\mrm{TC}}} - \mrm{e} ^ {-\iu t H_{\mrm{D}}}}\psi} 
        \leq \frac{\lambda}{\omega} \abs*{\sin(\omega t)} \min_\pm f^{\pm}_C(\psi_t) 
        + 3\frac{\lambda^2}{\omega} \int_{0}^{t }f_L( \psi_s ) \dd s
	\shortintertext{with \( f^\pm_{C},\, f_{L}\colon\:\denseD\to\R\) given by}
	\begin{aligned}
     &f^{\pm}_C(\psi) = \eof*{
        2\norm[\big]{\of{ {\spinalt{}}^2 - \Sz^2} \psi} 
        + 2 \norm[\big]{\of{\salt\pm}^2\psi}
            }^{\nicefrac{1}{2}}
        \norm[\big]{\of{\nnalt + 2 }\psi}^{\nicefrac{1}{2}}
	   \\
		&\!\begin{multlined}
			f_L(\psi) = 
			\biggl[	
				\norm[\big]{\of{\spinalt+ \spinalt-}^2 \psi}^{\nicefrac{1}{2}} 
				+ \norm[\big]{\of{\spinalt- \spinalt+}^2 \psi}^{\nicefrac{1}{2}}
				+ {\norm[\big]{\of{\spinalt+}^2\of{\spinalt-}^2 \psi}}^{\nicefrac{1}{2}}  
				\\
				+ {\norm[\big]{\of{\spinalt-}^2\of{\spinalt+}^2 \psi}}^{\nicefrac{1}{2}}
			\biggr]
			\cdot\norm[\big]{\of{\nnalt + 4}^2 \psi}^{\nicefrac{1}{2}}
		\end{multlined}
		\end{aligned}
	\end{gather}
	and where  \(\psi_t = U_2(t) \psi\), \(\spinalt\pm = \spin\pm\otimes\one\) and \(\nnalt = \one\otimes\nn\).
\end{prop}
We refer to the appendix for the proof; this relies on repeated applications of \cref{thm:bosonic-polynomial-bound,thm:partial-integration,eq:Cauchy--Schwarz-trick}, as detailed at page~\pageref{proof:intermediate-bound}.

Beyond being a key step toward the proof of \cref{thm:general-bound}, this result is also of independent interest. It can be used for \emph{a posteriori} error estimates when doing numerical simulations of the dynamics within the RWA. 
Once the dynamics of an initial state are being calculated using the RWA, the bound in \cref{thm:intermediate-bound} can be evaluated giving an estimate for the error introduced by incorporating the RWA without the need for also simulating the more complex Dicke dynamics. 
In addition, \cref{thm:intermediate-bound} provides an analytic upper bound on the RWA error for arbitrary states with finitely many photon excitations. However, it is not directly applicable for \textit{a priori} estimates, as it still involves the evolution group \(U_2(t)\), for which no simple closed-form expression is available. To proceed, one can either specialize to states for which the action of \(U_2(t)\) simplifies—i.e. the eigenstates of the Tavis--Cummings Hamiltonian introduced in \cref{ch:dicke-tc-hamiltonians}, for which only a global phase remains—or estimate the norms in \cref{eq:intermediate-bound} using operator norm bounds. The latter, more conservative strategy is discussed in \cref{ch:worst-case-bound}, while the former is pursued below and leads to our main result, \cref{thm:general-bound}.

\begin{proof}[{Proof of {\Cref{thm:general-bound}}}]
    We specialize \cref{thm:intermediate-bound} to the case where \(\psi\) is some eigenstate of the Tavis--Cummings Hamiltonian, i.e. \(\psi = \tcpsi{S}{M}{\sigma}\) for some \(S,M,\) and \(\sigma\) (cf.~\cref{eq:TC-eigenstates}).
    Then, due to \cref{eq:def-U2}, \(\psi_t = U_2(t)\psi = \e^{+\iu \omega t (M-S)}\e^{-\iu t \tceig{S}{M}{\sigma}} \tcpsi{S}{M}{\sigma}\) and we have 
    \(f^{\pm}_C(\psi_t) = f^{\pm}_C(\tcpsi{S}{M}{\sigma})\) and \(f_L(\psi_t) = f_L(\tcpsi{S}{M}{\sigma})\) due to linearity of all relevant operators.
    In particular, the integrand in \cref{eq:intermediate-bound} becomes independent of \(s\), whence
    \begin{align}
        \int_{0}^{t }f_L(\psi_s) \dd s
        = \int_{0}^{t } f_L(\tcpsi{S}{M}{\sigma}) \dd s
        = \abs{t} f_L(\tcpsi{S}{M}{\sigma})
        ,
    \end{align}
    which shows the bound in \cref{eq:general-bound}.

    To conclude the proof, we need to show that the limit~\eqref{eq:dicke-tc-convergence} holds for any $\psi\in\HH$.
    For linear combinations $\psi=\sum_{k=1}^{k_{\mrm{max}}} c_k \tcpsi{S_k}{M_k}{\sigma_k}$ of finitely many eigenstates, this immediately follows from \cref{eq:general-bound} by the triangle inequality.
    For a general \(\psi \in \HH\), we use the fact that the set of all eigenstates \(\tcpsi{S}{M}{\sigma}\) is a complete for \(\HH\) \cite{Bogoliubov_Bullough_Timonen_1996,bogoliubovTimeEvolutionAtomic2017}: given any \(\epsilon > 0\) there exist
    \begin{equation}
        \psi_\epsilon \in \Set*{\sum_{k=1}^{k_{\mrm{max}}} c_k \tcpsi{S_k}{M_k}{\sigma_k} \given c_k \in \C, k_{\mrm{max}} \in \N}
    \end{equation}
    with
    \begin{equation}
        \norm*{\psi - \psi_\epsilon} \leq \frac{\epsilon}{4}
    \end{equation}
    and \(\omega > 0\) with
    \begin{equation}
        \norm[\big]{\of[\big]{\e^{-\iu H_{\mathrm{D}}} - \e^{-\iu H_{\mathrm{TC}}}} \psi_\epsilon} \leq \frac{\epsilon}{2}
        ,
    \end{equation}
    because \(\norm[\big]{\of[\big]{\e^{-\iu H_{\mathrm{D}}} - \e^{-\iu H_{\mathrm{TC}}}} \psi_\epsilon}\to 0\) as \(\omega \to \infty\);
    but then, applying the triangle inequality,
    \begin{align}
        \norm[\big]{ \of[\big]{\mrm{e} ^ {-\iu t H_{\mrm{TC}}} - \mrm{e} ^ {-\iu t H_{\mrm{D}}}}\psi } 
        &\leq \norm[\big]{ \of[\big]{\mrm{e} ^ {-\iu t H_{\mrm{TC}}} - \mrm{e} ^ {-\iu t H_{\mrm{D}}}}\psi_\epsilon }+\norm[\big]{ \of[\big]{\mrm{e} ^ {-\iu t H_{\mrm{TC}}} - \mrm{e} ^ {-\iu t H_{\mrm{D}}}}(\psi-\psi_\epsilon) }
        \\
        &\leq\frac{\epsilon}{2} + \frac{\epsilon}{2} 
        = \epsilon
        ,
    \end{align}
    where we used the unitarity of \(\mrm{e} ^ {-\iu t H_{\mrm{TC}}}\) and \(\mrm{e} ^ {-\iu t H_{\mrm{D}}}\); since \(\epsilon\) was chosen arbitrarily, this proves \cref{eq:dicke-tc-convergence} and thus concludes the proof.
\end{proof}

We conclude this section with two technical remarks regarding the proof of \cref{thm:general-bound}, deferring the analysis of our error bounds to \cref{sec:analysis}.

\begin{remark}\label{thm:remark-reduction-special-case}
    The order in which we apply the \textit{Cauchy--Schwarz trick}~\eqref{eq:Cauchy--Schwarz-trick} and the bound for bosonic polynomials, \cref{thm:bosonic-polynomial-bound}, is relevant.
    If we had applied \cref{thm:bosonic-polynomial-bound} before \cref{eq:Cauchy--Schwarz-trick} in \cref{eq:dicke-tc-general-bound-proof3}, the resulting bound would only depend on the bosonic part of the state as the collective spin operators would have been bounded by their operator norm \(\opnorm{\s\pm} = N\).
    This alternative procedure would lead to a bound that directly reduces to the one proven in ref.~\cite{burgarthTamingRotatingWave2024} for the special case \(N=1\) and is directly applicable to arbitrary initial states in \(\denseD\).
    Instead, the bound proven here also depends on the spin state, as the operators are separated by \cref{eq:Cauchy--Schwarz-trick} first.
\end{remark}

\begin{remark}
    As already stated, in this analysis we restricted our attention to the resonant case $\omega_0=\omega$. An extension to our results to the general case of a nonzero detuning between the two energies, $\Delta=\omega_0-\omega$, is indeed possible at the price of additional expressions. 
    In this case the bound takes the form 
    \begin{gather}
		\label{eq:general-bound-of-resonant}
		\norm[\big]{ \of[\big]{\mrm{e} ^ {-\iu t H_{\mrm{TC}}} - \mrm{e} ^ {-\iu t H_{\mrm{D}}}}\tcpsi{S}{M}{\sigma} } 
        \!\begin{multlined}[t]
        \leq \frac{\lambda}{\omega}\abs*{\sin(\omega t)} \min_\pm f^{\pm}_C(\tcpsi{S}{M}{\sigma}) 
        \\
        +\abs{t}\left(\frac{\lambda}{\omega} \frac{\Delta}{2} f_d(\tcpsi{S}{M}{\sigma}) + 3\frac{\lambda^2}{\omega} f_L(\tcpsi{S}{M}{\sigma})\right)
        \end{multlined}
	\shortintertext{with \(f_{C}^\pm,f_d,f_{L}\colon\:\denseD\to\R\) given by}
	\begin{aligned}
     &f^{\pm}_C(\psi) = \eof*{
        2\norm[\big]{\of{ {\spinalt{}}^2 - \Sz^2} \psi} 
        + 2 \norm[\big]{\of{\salt\pm}^2\psi}
            }^{\nicefrac{1}{2}}
        \norm[\big]{\of{\nnalt + 2 }\psi}^{\nicefrac{1}{2}}
	   \\
		&\!\begin{multlined}
			f_L(\psi) = 
			\biggl[	
				\norm[\big]{\of{\spinalt+ \spinalt-}^2\psi}^{\nicefrac{1}{2}} 
				+ \norm[\big]{\of{\spinalt- \spinalt+}^2\psi}^{\nicefrac{1}{2}}
				+ {\norm[\big]{\of{\spinalt+}^2\of{\spinalt-}^2 \psi}}^{\nicefrac{1}{2}}  
				\\
				+ {\norm[\big]{\of{\spinalt-}^2\of{\spinalt+}^2 \psi}}^{\nicefrac{1}{2}}
			\biggr]
			\cdot\norm[\big]{\of{\nnalt +4}^2 \psi}^{\nicefrac{1}{2}}
		\end{multlined}
        \\
        &f_d(\psi) = \eof*{\norm[\big]{\salt+\salt-\psi}^{\nicefrac{1}{2}} +\norm[\big]{\salt-\salt+\psi}^{\nicefrac{1}{2}}}\norm[\big]{(\nnalt+1)\psi}^{\nicefrac{1}{2}}
        .
		\end{aligned}
	\end{gather}
    The proof is analogous to that of \cref{thm:general-bound}, so we omit it here and focus our discussion on the resonant case.
\end{remark}

\subsection{Alternative worst-case bound}
\label{ch:worst-case-bound}
As mentioned in the discussion of \cref{thm:intermediate-bound}, an alternative strategy for removing the dependence on the evolution group \(U_2\) in \cref{eq:intermediate-bound} is available. Here, we present that approach. The idea is to estimate the norms on the right-hand side of \cref{eq:intermediate-bound}, i.~e. the expressions
\begin{align}
    \label{eq:necessary-estimate-1}
    &\norm[\big]{ \of{\salt\pm}^2 \of{\salt\mp}^2U_2(t)\psi }
    ,
    \\
    \label{eq:necessary-estimate-2}
    &\norm[\big]{\of[\big]{\salt\pm\salt\mp}^2U_2(t)\psi}
    ,
    \\
    \label{eq:necessary-estimate-3}
    &\norm[\Big]{\of[\big]{\salt\pm}^2 U_2(t)\psi }
    ,
    \\
    \label{eq:necessary-estimate-4}
    &\norm[\Big]{\of[\big]{{\salt{}}^2- \boldsymbol{S}_z^2}U_2(t)\psi }
    ,
    \\
    \label{eq:necessary-estimate-5}
    &\norm[\big]{\of{\nnalt + 2 } U_2(t)\psi}
    ,
    \shortintertext{and}
    \label{eq:necessary-estimate-6}
    &\norm[\big]{\of{\nnalt + 4}^2 U_2(t)\psi}
    ,
\end{align}
with time-independent quantities.
For the expressions in \eqref{eq:necessary-estimate-5} and \eqref{eq:necessary-estimate-6} we make use of the conservation of \(H_0\) under \(U_2(t)\).
Note that, in the following, we will use the following notation: 
for two symmetric operators \(A\colon \dom(A)\subset \HH \to \HH\) and \(B \colon \dom(B) \subset \HH \to \HH \), we write
\begin{equation}
    A \leq B \ratio\Leftrightarrow \innerp{\psi}{A\psi} \leq \innerp{\psi}{B\psi} ,\quad \forall \psi \in \dom(A)\cap \dom(B).
\end{equation}
    
\begin{lemma}\label{thm:number-and-U2}
	Let $\nnalt = \one\otimes\nn$, with \(\nn\) the number operator~\eqref{eq:number-operator-def}. 
    Then
	\begin{equation}\label{eq:number-and-U2}
		U_2\adj(t)\, \nnalt\, U_2(t) \leq \of*{\nnalt + N\one\otimes\one}.
	\end{equation}
 \end{lemma}
\begin{proof}
    Recall that, given an orthonormal projector \(P = P\adj = P^2\) and some densely defined operator \(A\) on a Hilbert space \(\HH\), we always have 
	\begin{align}
		P - A\adj P A &\leq \one, \label{eq:number-and-U2-part-1}
	\end{align}
	as a direct consequence of the inequalities \(0\leq P \leq \one\). 
	Similarly, for every unitary operator \(U\) on \(\HH\),
	\begin{align}
		U\adj P U - P &\leq \one. \label{eq:number-and-U2-part-2}
	\end{align}
    In our case, the matrices \(\pauliPlusI{i}\pauliMinusI{i}\) and \(\pauliMinusI{i}\pauliPlusI{i}\)
	are orthonormal projectors.
	As such, recalling that the free Hamiltonian \(H_0\) is a conserved quantity for \(U_{2}(t)\) by \cref{thm:U2-H0-commutation}, we have	
	\begin{align}
		U_2\adj(t) \nnalt U_2(t) 
		&= \frac{1}{\omega} U_2\adj(t) H_0 U_2(t) - U_2\adj(t) \SZalt U_2(t)	
		\\
		&= \frac{1}{\omega} H_0 - U_2\adj(t) \SZalt U_2(t)	
		\\
		&= \nnalt + \of*{\SZalt - U_2\adj(t) \SZalt U_2(t)}	
		\\
		&\!\begin{multlined}
			= \nnalt + \frac{1}{2}\sum_{i=1}^{N}
			\biggl[
				\underbrace{
					\pauliPlusI{i}\pauliMinusI{i}\otimes\one - U_2\adj(t) \pauliPlusI{i}\pauliMinusI{i}\otimes\one U_2(t)
				}_{\leq \one,\text{ by \cref{eq:number-and-U2-part-1}}}
			\\
				+ \underbrace{
				U_2\adj(t) \pauliMinusI{i}\pauliPlusI{i}\otimes\one U_2(t) -\pauliMinusI{i}\pauliPlusI{i}\otimes\one
				}_{\leq \one,\text{ by \cref{eq:number-and-U2-part-2}}} 
			\biggr]
		\end{multlined}
		\label{eq:number-and-U2-part-3}
		\\
		&\leq \nnalt + N \one\otimes\one ,
	\end{align}
	where, in \cref{eq:number-and-U2-part-3}, we have used \cref{eq:ladder-algebra}.
\end{proof}

\begin{prop}\label{thm:worst-case-bound}
    For every \(\psi\in\denseD\), the following estimate holds for any $t>0$:
	\begin{multline}
        \label{eq:worst-case-bound}
        \norm[\big]{ \of[\big]{\mrm{e} ^ {-\iu t H_{\mrm{TC}}} - \mrm{e} ^ {-\iu t H_{\mrm{D}}}}\psi} 
        \leq
        \frac{\lambda}{\omega}\abs*{\sin\of*{\omega t}} \sqrt{N^2 + N} \sqrt{\norm[\big]{\of*{\nnalt + N + 2}} \psi}
        \\
        + 3 \frac{\lambda^2}{\omega} \abs{t} N^2 \sqrt{\norm[\big]{ \of*{\nnalt + N + 4}^2\psi}}
        .
	\end{multline}
\end{prop}
\begin{proof}
    We apply \cref{thm:intermediate-bound}.
    For this, we need to find suitable estimates for the expressions in \cref{eq:necessary-estimate-1,eq:necessary-estimate-2,eq:necessary-estimate-3,eq:necessary-estimate-4,eq:necessary-estimate-5,eq:necessary-estimate-6}.
    Due to \cref{eq:number-and-U2} we already have such estimates at hand for the expressions in \cref{eq:necessary-estimate-5,eq:necessary-estimate-6}:
    \begin{align}
        \label{eq:necessary-estimate-5-proven}
        \norm[\big]{\of{\nnalt + 2 } U_2(t) \psi}
        &\leq \norm[\big]{\of{\nnalt + 2 + N }\psi}
        \\
        \label{eq:necessary-estimate-6-proven}
        \norm[\big]{\of{\nnalt + 4}^2 U_2(t) \psi}
        &\leq \norm[\big]{\of{\nnalt + 4 + N}^2 \psi}
        .
    \end{align}
    For the remaining four expressions we shall use the properties \(\opnorm[\big]{\spin\pm} = {N}/{2} \), and \(\opnorm[\big]{{\spin{}}^2 - S_z^2} \leq {N^2}/{4} + {N}/{2}\).
    Thus, for \(\psi \in \denseD\), using the sub-multiplicativity of \(\opnorm{}\), the unitarity of \(U_2(t)\), and the triangle-inequality, we get
    \begin{align}
        \norm[\big]{ \of{\salt\pm}^2 \of{\salt\mp}^2U_2(t)\psi }
        &
        \leq \frac{N^4}{2^4} \norm*{\psi}
        ,
        \\
        \norm[\big]{\of[\big]{\salt\pm\salt\mp}^2U_2(t)\psi}
        &
        \leq \frac{N^4}{2^4} \norm*{\psi}
        ,
        \\
        \norm[\Big]{\of[\big]{\salt\pm}^2 U_2(t)\psi }
        &
        \leq \frac{N^2}{2^2}\norm*{\psi}
        ,
        \\
        \norm[\Big]{\of[\big]{{\salt{}}^2- \boldsymbol{S}_z^2}U_2(t)\psi }
        &
        \leq \of*{\frac{N^2}{4} + \frac{N}{2}}\norm*{\psi}
        .
    \end{align}
   Inserting these estimates together with \cref{eq:necessary-estimate-5-proven,eq:necessary-estimate-6-proven} into \cref{eq:intermediate-bound} gives \cref{eq:worst-case-bound} and thus concludes the proof.
\end{proof}
While simple and applicable to arbitrary states in \(\denseD\), \cref{thm:worst-case-bound} contains no information about the spin component of the initial state beyond the number of spins \(N\).
The only state-dependent quantity entering the bound is the number of field excitations, as the operator norms used in the estimates~\cref{eq:necessary-estimate-1,eq:necessary-estimate-2,eq:necessary-estimate-3,eq:necessary-estimate-4} involve a supremum over all spin states and thus reflect a worst-case scenario. 
To analyse the dependence of the RWA error on the spin component of the initial state, \cref{thm:general-bound} is preferable. However, for generalizing our results to other models with a structure similar to the Dicke model, the worst-case bound in \cref{thm:worst-case-bound} is better suited, as it does not rely on the existence of a complete set of eigenstates.

\section{Analysis of the error bound}\label{sec:analysis}

In this section, we analyse the error bound proven in \cref{ch:general-bound} by further evaluating the individual components of \cref{eq:general-bound} and examining their dependence on the model parameters and the initial state in \cref{ch:specific-bounds}. We then compare these analytical bounds with numerical results in \cref{ch:numerics}, focusing on how the observed behaviour reflects the identified parameter dependencies.

\subsection{Analytic discussion of the properties of the bound}
\label{ch:specific-bounds}
By writing the normalized eigenstates 
\begin{align}
    \tcpsinormed{S}{M}{\sigma} 
    = \frac{\tcpsi{S}{M}{\sigma}}{\norm{\tcpsi{S}{M}{\sigma}}}
    = \sum_{n=M-2S}^{M} \tilde c_{S,M,n}^\sigma \ket{S, M-n-S}\otimes\phi_n
\end{align}
of the Tavis--Cummings Hamiltonian in terms of products \(\ket{S,m}\otimes \phi_n\) of Dicke and Fock states, we can calculate the state dependences in \cref{eq:general-bound} explicitly and determine the dependency of the RWA error on the parameters of the system and initial state.
For this purpose, we need to apply the following estimate:
\begin{align}\label{eq:normalized-coefficients-approximation}
    \abs{\tilde c_{S,M,n} ^\sigma}^2 
    \leq 1
\end{align}
because the coefficients \(\tilde c_{S,M,n}\) depend on the particular solution \(\tcparvec{\sigma}\) of the Bethe equations~\eqref{eq:TC-bethe-equations} in a nontrivial way.

\begin{corollary}\label{thm:scaling-S-M-bound}
    Let \(\tcpsinormed{S}{M}{\sigma} = \tcpsi{S}{M}{\sigma}/\norm{\tcpsi{S}{M}{\sigma}}\) be a normalized eigenstates of the Tavis--Cummings Hamiltonian. Then the difference between the Dicke and Tavis--Cummings dynamics satisfies
    \begin{equation}
        \norm[\big]{ \of[\big]{\mrm{e} ^ {-\iu t H_{\mrm{TC}}} - \mrm{e} ^ {-\iu t H_{\mrm{D}}}}\tcpsinormed{S}{M}{\sigma}}
        \label{eq:scaling-S-M-bound}
        \leq 2\frac{\lambda}{\omega} \abs*{\sin(\omega t)}\of*{S+1}^2 \sqrt{M + 2}
        + 18\frac{\lambda^2}{\omega} \abs{t} \of*{S + 1}^{\nicefrac{7}{2}} \of*{M + 2}
        .
    \end{equation}
\end{corollary}
The proof of this result is in the appendix at page~\pageref{proof:scaling-S-M-bound}. Note that the obtained bound does not directly depend on the number of spins \(N\).

As the set of all eigenstates \(\tcpsi{S}{M}{\sigma}\) of the Tavis--Cummings Hamiltonian is a complete orthogonal set for \(\HH\)~\cite{Bogoliubov_Bullough_Timonen_1996,bogoliubovTimeEvolutionAtomic2017} we can in particular write any state \(\psi \in\denseD\) as a linear combination of eigenstates \(\tcpsi{S}{M}{\sigma}\).
This means that the same scaling with the total angular momentum and the total number of excitations found in \cref{thm:scaling-S-M-bound} also applies to more general states in the finite photon space \(\denseD\), as long as
they can be represented by a \emph{finite} liner combination of the eigenstates.
\begin{corollary}\label{thm:scaling-S-M-bound-extended}
	Let \(\psi \in \denseD\) be a normalized state in the form
    \begin{align}
		\psi = 
        \sum_{S\primed \in \mathcal{M}_S} \sum_{M\primed\in \mathcal{M}_M}\sum_{\sigma\primed \in \mathcal{M}_\sigma} \alpha_{S\primed,M\primed,\sigma} \tcpsinormed{S\primed}{M\primed}{\sigma}
	\end{align}
	for \(\alpha_{S\primed,M\primed,\sigma\primed} \in \C\setminus\cof{0}\) and with \(S = \max \mathcal{M}_S\), \(M = \max \mathcal{M}_M\).
    Then
    \begin{align}
        \norm[\big]{ \of[\big]{\mrm{e} ^ {-\iu t H_{\mrm{TC}}} - \mrm{e} ^ {-\iu t H_{\mrm{D}}}}\psi}
        \!\begin{multlined}[t]
        \leq \abs*{\mathcal{M}_S}\abs*{\mathcal{M}_M}\abs*{\mathcal{M}_\sigma} 
        \left(\vphantom{\frac{\lambda^2}{\omega}}2\frac{\lambda}{\omega} \abs*{\sin(\omega t)}\of*{S+1}^2 \sqrt{M + 2}\right.
        \\
        \left. + 18\frac{\lambda^2}{\omega} \abs{t} \of*{S + 1}^{\nicefrac{7}{2}} \of*{M + 2}\right)
        ,
        \end{multlined}
    \end{align}
    where \(\abs*{\mathcal{M}_\#}\) denotes the number of elements in the set \(\mathcal{M}_\#\).
\end{corollary}

\begin{proof}
	This is an immediate consequence of \cref{thm:scaling-S-M-bound}, the triangle inequality, and the estimates \(S\primed \leq S\) and \(M\primed \leq M\).
\end{proof}

This result applies, for instance, to Dicke states: we have
\begin{align}
	\ket{S,m} \otimes \phi_n
	= \sum_{\sigma=1}^{K} \tilde c_{S,S+n+m,n}^\sigma \tcpsinormed{S}{n+S+m}{\sigma}
	,\quad K = {\min\of*{2S,S+n+m}+1}
    ,
\end{align}
whence, by \cref{thm:scaling-S-M-bound-extended},
\begin{multline}\label{eq:scaling-S-M-bound-dicke-fock}
	\norm[\big]{ \of[\big]{\mrm{e} ^ {-\iu t H_{\mrm{TC}}} - \mrm{e} ^ {-\iu t H_{\mrm{D}}}}\ket{S,m}\otimes\phi_n}
        \leq K\left(\vphantom{\frac{\lambda^2}{\omega}}2\frac{\lambda}{\omega} \abs*{\sin(\omega t)}\of*{S+1}^2 \sqrt{S+n+m + 2}\right.
        \\
        \left. + 18\frac{\lambda^2}{\omega} \abs{t} \of*{S + 1}^{\nicefrac{7}{2}} \of*{S+n+m + 2}\right)
\end{multline}
which underpins the dependence of the bound on the magnetic number \(m\) for such states. 

Let us comment on \cref{thm:scaling-S-M-bound,thm:scaling-S-M-bound-extended}, respectively.
In contrast to the worst--case bound presented in \cref{thm:worst-case-bound}, here we see that, instead of the mere number of photons, the \emph{total} number of excitations, \(M\), which includes the excitations in the spin-component of the initial state, plays a central role for the validity of the RWA.
The bound~\eqref{eq:scaling-S-M-bound} indicates that the RWA becomes less accurate with a rising total number of excitations---regardless of these excitations originating in the field or in the spin component of the system.
In addition, in \cref{thm:scaling-S-M-bound} the role of the number of spins \(N\) is being replaced by the total angular momentum \(S \leq N/2\): the quality of the RWA depends crucially on the total angular momentum in the initial state.

This has important consequences. Even when a large number of spins collectively interact with the field, the initial state may have low total angular momentum \(S\), in which case the rotating-wave approximation (RWA) remains valid—as long as the other parameters, \(\omega\), \(\lambda\), and \(M\), lie in appropriate regimes. The only role of \(N\) in this context is to set an upper bound on the possible values of \(S\). As a result, the RWA can still be considered valid in the joint limit of high frequency \(\omega \to \infty\) and large particle number \(N \to \infty\), the latter corresponding to the thermodynamic limit:

\begin{corollary}
    \label{thm:thermodynamic-limit-convergence}
    Let \(M \in N\) be fixed and \((S_i)_{i\in \N}\subset \N/2\), \((N_i)_{i \in \N}\subset \N\), \((\lambda_i)_{i \in \N}\subset \R\), and \((\omega_i)_{i \in \N}\subset \R\) be sequences with \(N_i/2 \geq S_i\) for every \(i\in \N\).
    
    If the sequences \(S_i\), \(\lambda_i\) and \(\omega_i\) satisfy
    \begin{equation}
        \label{eq:thermodynamic-limit-convergence-condition}
        {S_i^{\nicefrac{7}{2}}} \lambda_i^2 / {\omega_i} \to 0 
    \end{equation}
    we have
    \begin{equation}\label{eq:thermodynamic-limit-convergence}
        \norm[\big]{ \of[\big]{\mrm{e} ^ {-\iu t H_{\mrm{TC}}(N_i, \omega_i)} - \mrm{e} ^ {-\iu t H_{\mrm{D}}(N_i, \omega_i)}}\tcpsinormed{S_i}{M}{\sigma}(N_i)} \to 0
    \end{equation}
    for \(i \to \infty\), where we made the dependences of \(H_{\mrm{D}}\), \(H_{\mrm{TC}}\), and \(\tcpsi{S}{M}{\sigma}\) on \(N\) and \(\omega\) explicit. 
\end{corollary}
\begin{proof}
    This is a direct consequence of \cref{thm:scaling-S-M-bound}.
\end{proof}

Importantly, we have the convergence in \cref{eq:thermodynamic-limit-convergence} if \(S_i = S\) is fixed and \(N_i \to \infty\) as well as \(\omega_i\to \infty\), i.~e. for initial states with fixed total angular momentum the RWA is valid also within the thermodynamic limit \(N \to \infty\). 
On the other hand, if \(S_i = N_i/2\) is maximal for each \(N_i\) and we rescale the coupling strength \(\lambda_i \mapsto \lambda_i/\sqrt{N_i}\), as common in the study of the thermodynamic limit~\cite{Hepp_Lieb_1973}, instead of \cref{eq:thermodynamic-limit-convergence-condition} we have the following condition for convergence:
\begin{equation}
    \label{eq:thermodynamic-limit-convergence-rescaled}
    {N_i^{\nicefrac{5}{2}}} \lambda_i^2 / {\omega_i} \to 0 
    \quad \text{and}\quad
    {N_i^{\nicefrac{3}{2}}} \lambda_i / {\omega_i} \to 0
    .
\end{equation}
This imposes a condition on the relation of the rates at which \(N\) and \(\omega\) jointly approach infinity that is sufficient for the validity of the RWA within the thermodynamic limit.

\subsection{Comparison with numerical results}\label{ch:numerics}

\begin{figure}[t]
    \centering
    \begin{subfigure}[][][t]{0.49\linewidth}
        \caption{}
        \centering
        \includegraphics[width=\linewidth]{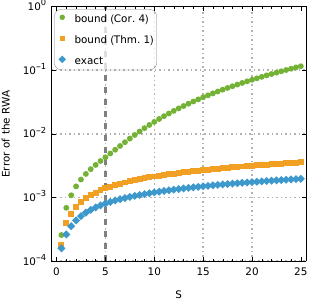}
        \label{fig:S-dependence}
    \end{subfigure}
    \hfill
    \begin{subfigure}[][][t]{0.49\linewidth}
        \caption{}
        \centering
        \includegraphics[width=\linewidth]{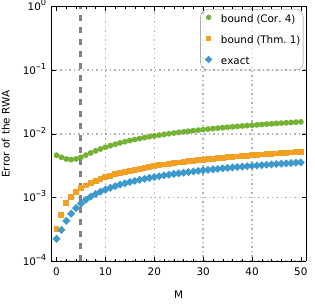}
        \label{fig:M-dependence}
    \end{subfigure}
    \caption{
    Dependence of the RWA error \(\norm*{\of*{\e^{-\iu t H_{\mathrm{D}}} - \e^{-\iu t H_{\mathrm{TC}}}}\tcpsi{S}{M}{\sigma}}\) on (a) the total angular momentum quantum number \(S\), and (b) the total number of excitations \(M\).
    Plotted are the exact numerical error (\cref{eq:exact-error-numerical}, blue rhombuses), along with the corresponding values of the error bounds from \cref{thm:general-bound} (orange squares) and \cref{thm:scaling-S-M-bound} (green circles), using a logarithmic scale on the vertical axis.
    The dashed vertical line marks the reference value \(M = 5\) or \(S = 5\) used in the complementary panel.
    Other parameters are \(\omega = 3000\), \(\lambda = 0.3\), and \(t = \frac{\pi}{4\omega}\).
    }
    \label{fig:S-M-dependence}
\end{figure}

\begin{figure}[t]
    \centering
    \begin{subfigure}[][][t]{0.49\linewidth}
        \caption{}
        \centering
        \includegraphics[width=\linewidth]{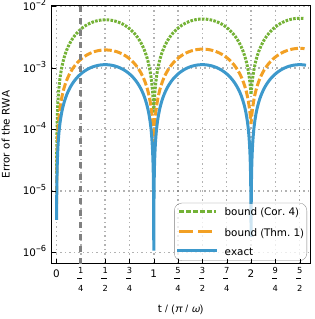}
        \label{fig:t-dependence}
    \end{subfigure}
    \hfill
    \begin{subfigure}[][][t]{0.49\linewidth}
        \caption{}
        \centering
        \includegraphics[width=\linewidth]{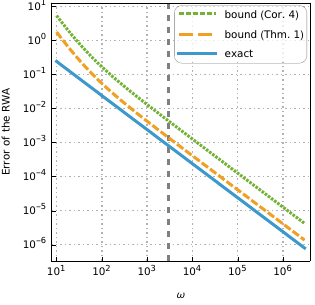}
        \label{fig:freq-dependence}
    \end{subfigure}
    \caption{
    Dependence of the RWA error \(\norm*{\of*{\e^{-\iu t H_{\mathrm{D}}} - \e^{-\iu t H_{\mathrm{TC}}}}\tcpsi{S}{M}{\sigma}}\) on (a) the evolution time \(t\) (in units of \(\pi/\omega\)) and (b) the field frequency \(\omega\).
    Plotted are the exact error (\cref{eq:exact-error-numerical}, solid blue), the error bound from \cref{thm:general-bound} (dashed orange), and the analytical bound from \cref{thm:scaling-S-M-bound} (dotted green). 
    The vertical axis uses a logarithmic scale in both panels; in panel~(b), the horizontal axis is logarithmic as well.
    A dashed vertical line marks the value \(t = \tfrac{1}{4} \cdot (\pi/\omega)\) in (a) and \(\omega = 3000\) in (b), consistent with other plots.
    Parameters: \(\lambda = 0.3\), \(S = 5\), \(M = 5\).
    }
    \label{fig:t-w-dependence}
\end{figure}

To assess whether the aforementioned dependencies on the total angular momentum quantum number \(S\) and the total number of excitations \(M\) are also reflected in the actual dynamics, we performed numerical simulations using \emph{Mathematica}~\cite{Mathematica}.

We begin by outlining our numerical procedure. Starting from an initial guess \(\tcparvec{}_{0} = (a + \iu b, a + 2 \iu b, \dots, a + M \iu b)\), with \(a = \sqrt{2S}/2\) and \(b = \sqrt{2S}/M\), we solve the Bethe equations~\eqref{eq:TC-bethe-equations} for the solution \(\tcparvec{\sigma}\) closest to \(\tcparvec{}_{0}\). This particular choice of initial point was motivated by its effectiveness in previous work by Bogoliubov et al.~\cite{bogoliubovTimeEvolutionAtomic2017}. 
Once a solution \(\tcparvec{\sigma}\) is obtained, we construct a finite-dimensional approximation \(\tilde\psi_{S,M,n_{\mathrm{max}}}^{\sigma}\) of the corresponding normalized eigenvector \(\tcpsinormed{S}{M}{\sigma}\), using \cref{eq:TC-eigenstates} and restricting to the subspace spanned by \(\ket{S,m} \otimes \phi_n\) for \(\abs{m} \leq S\) and \(n \leq n_{\mathrm{max}} = 2(M + 2S)\). The state is normalized via \cref{eq:tc-eigenstates-norm,eq:tc-eigenstates-coefficients}, using the specific Bethe solution \(\tcparvec{\sigma}\).

Using the same basis, we construct finite-dimensional approximations \(H_{\mrm{D},n_{\mathrm{max}}}\) and \(H_{\mrm{TC},n_{\mathrm{max}}}\) of the Dicke and Tavis--Cummings Hamiltonians, respectively. We then compute the corresponding matrix exponentials and evaluate the norm difference
\begin{align}
    \label{eq:exact-error-numerical}
    \norm[\big]{ \of[\big]{\mrm{e} ^ {-\iu t H_{\mrm{TC},n_{\mathrm{max}}}} - \mrm{e} ^ {-\iu t H_{\mrm{D},n_{\mathrm{max}}}}}\tilde\psi_{S,M,n_{\mathrm{max}}}^{\sigma}}
    .
\end{align}

We compare the numerical results with the error bound stated in \cref{thm:general-bound}, evaluated in two different ways.  
The first approach is numerical: we compute the right-hand side of \cref{eq:general-bound} using the same finite-dimensional approximation \(\tilde\psi_{S,M,n_{\mathrm{max}}}^{\sigma}\) employed in the simulation of the exact dynamics. 
The second approach follows the ideas developed in \cref{thm:scaling-S-M-bound}. Specifically, it involves estimating the bound in \cref{eq:general-bound} using the approximation for the normalized coefficients provided in \cref{eq:normalized-coefficients-approximation}, without relying on the explicit form of the eigenstate \(\tcpsi{S}{M}{\sigma}\). Note that we do not use the final expression in \cref{eq:scaling-S-M-bound} for the comparison, as it involves further simplifications intended for analytical transparency rather than numerical tightness. Instead, the version we compare is the intermediate expression obtained by inserting \cref{eq:f-C-explicit-eval,eq:f-L-explicit-eval} into the general bound from \cref{thm:general-bound}. In summary, the first evaluation method yields a tighter bound, while the second provides a fully analytic estimate.

\Cref{fig:S-M-dependence} displays the numerical results for varying values of \(S\) and \(M\), alongside the analytic error bounds from \cref{thm:general-bound} and \cref{thm:scaling-S-M-bound}.  
The bound provided by \cref{thm:general-bound} consistently overestimates the exact error by an approximately constant factor.  
A separate analysis suggests that this factor converges to a value in the range \([1.4, 1.9]\) as either \(S\) or \(M\) increases.  
Both the \(S\)-dependence and \(M\)-dependence of the error are accurately reflected in the bound of \cref{thm:general-bound}.

When estimating the bound analytically rather than through direct numerical evaluation—i.e., when using the bound from \cref{thm:scaling-S-M-bound}—the deviation from the exact error increases due to the additional approximation introduced in \cref{eq:normalized-coefficients-approximation}.  
In this case, the dependence on \(M\) is still well captured up to an approximately constant factor, which appears to approach a value around \(4.3\).  
However, the relative deviation grows gradually with increasing \(S\), indicating a reduced accuracy of the bound in that regime.  
Nevertheless, in both the analytical and numerical evaluations, the overall qualitative behaviour—namely, the growth of the RWA error with increasing \(S\) and \(M\)—is consistently reflected.  
Unlike \cref{thm:scaling-S-M-bound}, the direct evaluation of the bound in \cref{thm:general-bound}, as well as the exact numerical error, demonstrate that the dependence of the RWA error on \(S\) is of comparable magnitude to its dependence on \(M\).

\Cref{fig:t-w-dependence} presents a similar comparison, this time focusing on the dependence of the RWA error on the evolution time \(t\) and the field frequency \(\omega\), while keeping \(S\) and \(M\) fixed.  
The exact error displays characteristic oscillations, with local minima occurring at times \(t \approx n\pi/\omega\) and local maxima near \(t \approx n\pi/(2\omega)\) for \(n \in \N\).  
Both error bounds—the one from \cref{thm:general-bound} and the analytical version from \cref{thm:scaling-S-M-bound}—accurately reflect this oscillatory behaviour.  
While the bound from \cref{thm:scaling-S-M-bound} is looser due to additional approximations, it still captures the correct functional dependence on time.  
The same applies to the dependence on \(\omega\): both bounds reproduce the correct convergence behaviour as \(\omega \to \infty\), albeit with an offset given by a constant multiplicative factor relative to the exact error.

\section{Conclusion and outlooks}\label{ch:conclusion}

In this paper, we provided a rigorous proof to the validity of the rotating-wave approximation for the Dicke model with $N\geq1$ spins, generalizing previous results for the specific case $N=1$. We have shown that, similarly as in the single-spin case, the RWA converges in the \emph{strong dynamical sense}---the difference between the evolutions of any state generated by the Dicke model and its rotating-wave approximation (namely, the Tavis--Cummings model) converges to zero for any time $t$ as the field frequency $\omega$ goes to infinity. To the best of our knowledge, this serves as the first rigorous proof of the validity of RWA for the Dicke model. Besides, the convergence of the RWA in the thermodynamic limit $N\to\infty$ was also proven for a specific class of states. 

Furthermore, our analysis is fundamentally \emph{quantitative}: for sufficiently well-behaved initial states—specifically, those whose field component contains a finite number \(n\) of photons—we explicitly compute the norm difference between the exact and approximate evolutions. This enables a detailed investigation of how the RWA error depends on both the model parameters—most notably, the number \(N\) of spins—and the properties of the initial state. A key observation is that the error depends not on \(N\) itself, but only on the total angular momentum \(S \leq N/2\) encoded in the initial state. This is especially relevant in regimes where direct numerical simulations become infeasible, such as for large \(N\). In addition, our methods capture essential features of the \emph{time dependence} of the RWA error, which turns out to be non-monotonic.

Also, the fundamental \emph{state-dependent} nature of the bound and thus of the convergence rate of the RWA---a phenomenon typical of infinite-dimensional Hilbert spaces---here acquires a particular importance: different states with different physical properties can behave in a starkly different way in relation to the validity of the RWA.
Again, a good example is the case of large $N$: the rate at which the ratio of \(\omega\) and \(N\) needs to approach $0$ to have convergence of the RWA depends on the specific choice of the initial state.

On the technical side, the error bounds are derived using a version of the integration-by-parts method introduced in~\cite{burgarthOneBoundRule2022}, adapted here to unbounded operators—similarly to~\cite{burgarthTamingRotatingWave2024}, which treated the special case \(N = 1\). Unlike that work, which relied on an explicit computation of the Jaynes--Cummings dynamics—an approach that becomes intractable for \(N > 1\) or for more general models—we adopt a new strategy that avoids such calculations. Instead, we exploit the existence of a complete set of eigenstates of the Tavis--Cummings Hamiltonian. 
While these eigenstates are known in terms of solutions to ordinary equations, our method does not rely on their explicit form. 
This opens the door to similar treatments in more abstract settings, leveraging the additional symmetry introduced by the RWA. 
In particular, the worst-case bound discussed in \cref{ch:worst-case-bound} can be readily extended to more complex models that share structural similarities with the Dicke model. Indeed, this bound does not rely on the existence of a complete set of eigenstates, making it amenable to more abstract and general arguments.
Crucially, our abstract results only require that the initial state lies in a suitable invariant domain: the space of states with a photon component containing a finite (though arbitrarily large) number of excitations. 
This choice was key to overcoming the technical challenges posed by the multi-spin case and, we believe, provides a natural foundation for proving the validity of the RWA—and other approximations—in more complex models of matter–field interaction relevant to applications.

In this regard, a natural research direction would be analysing the behaviour of the RWA in the case of non-monochromatic---possibly even continuous---boson fields. 
In such a scenario, one should compare the evolution groups generated by the multimode generalizations of the Rabi model \cite{fannes1988equilibrium,amann1991ground,hirokawa1999expression,arai1990asymptotic,hubner1995spectral,bach2017existence,reker2020existence,dam2020asymptotics,hasler2021existence}) and the Jaynes--Cummings \cite{hirokawa2001remarks,jakvsic2006mathematical,lonigroGeneralizedSpinbosonModels2022} or Tavis--Cummings one \cite{,lonigroSelfAdjointnessClassMultiSpinBoson2023}. 
Such cases, however, are expected to exhibit novel complications either due to the multimode nature of the boson field---which makes the very identification of the correct rotating frame nontrivial \emph{a priori}---or to the possible presence of infrared or ultraviolet divergences \cite{bach2017existence,dam2020asymptotics,lonigroSelfAdjointnessClassMultiSpinBoson2023} that could cause a breakdown of the RWA.    
Future research will be devoted to such generalizations.

A comparison between the error bounds and the exact numerical error shows that the qualitative behaviour of the bounds closely mirrors that of the actual error. In particular, the dependence on the total number of excitations in the initial state, the evolution time, and the field frequency is well captured quantitatively. Based on this agreement, we conclude that even for large values of \(S\) and \(M\)—where direct numerical simulation is impractical—the RWA error is expected to grow with \(S\) and \(M\) in a similar manner.

Applying more sophisticated numerical methods could strengthen some of the heuristic conclusions made towards the end of \cref{ch:specific-bounds}.
Analytically finding lower bounds on the error of the RWA for the Dicke model, similarly to what was obtained in ref.~\cite{burgarthTamingRotatingWave2024}, could potentially give a rigorous argument for the bounds mirroring the actual dynamics. 
Moreover, relating the spin operators appearing in our bound to a suitable entanglement measure could provide an even deeper understanding of the connection between the validity of the RWA and the initial spin-state.

\section*{Acknowledgments}

The authors thank Robin Hillier for fruitful discussions.
D.L. acknowledges financial support by Friedrich-Alexander-Universit\"at Erlangen-N\"urnberg through the funding program ``Emerging Talent Initiative'' (ETI), and was partially supported by the project TEC-2024/COM-84 QUITEMAD-CM.

\appendix 
\setcounter{equation}{0}
\renewcommand{\theequation}{A.\arabic{equation}}
\section*{Appendix}
In this appendix, we provide the proofs that were omitted from the main text.

\begin{proof}[{Proof of \cref{thm:intermediate-bound}}]\label{proof:intermediate-bound}
	Let \(\psi\in\denseD\) and \(\psi_t = U_2(t)\psi\) for some \(t\in\R\).
	By the triangle inequality and \cref{thm:partial-integration}, we have
	\begin{equation}\label{eq:dicke-tc-general-bound-proof1}
		\norm{\of{U_2(t) - U_1(t)}\psi} 
        \leq \norm{S_{21}(t)\psi_t} + \int_0^t \norm{\of{S_{21}(s)H_2 - H_1(s) S_{21}(s)}\psi_t} \dd s.
	\end{equation}
	We will bound both terms of \cref{eq:dicke-tc-general-bound-proof1} separately.
	For the first term of \cref{eq:dicke-tc-general-bound-proof1}, we use \cref{eq:integrated-action} for the integrated action \(S_{21}(t)\), and distribute the addends in the inner product:
	\begin{align}
		\norm{S_{21}(t)\psi_t}^2 
		&= \innerp{S_{21}(t)\psi_t}{S_{21}(s)\psi_t}
		\\
		&=\begin{aligned}[t]
			\frac{\lambda^2}{\omega^2}\sin^2(\omega t) 
			\biggl[
			&\innerp[\Big]{ \mrm{e} ^ {+\iu\omega t} (\spin+\otimes\ad) \psi_t }{ \mrm{e} ^ {+\iu\omega t}(\spin+\otimes\ad) \psi_t }
			\\
 			&+\innerp[\Big]{ \mrm{e} ^ {-\iu\omega t} (\spin-\otimes\an) \psi_t }{ \mrm{e} ^ {-\iu\omega t}(\spin-\otimes\an) \psi_t }
			\\
			&+\innerp[\Big]{ \mrm{e} ^ {+ \iu\omega t} (\spin+\otimes\ad) \psi_t }{ \mrm{e} ^ {-\iu\omega t}(\spin-\otimes\an) \psi_t }
			\\
			&+\innerp[\Big]{ \mrm{e} ^ {-\iu\omega t} (\spin-\otimes\an) \psi_t }{ \mrm{e} ^ {+\iu\omega t}(\spin+\otimes\ad) \psi_t }
			\biggr],
		\end{aligned}
		\intertext{
			collect operators on one side of the inner-product in the first two summands, and identify \(\Re(z)= 2^{-1}(z+\bar z)\) in the last two summands:
		}
		&=\begin{aligned}[t]
		\frac{\lambda^2}{\omega^2}\sin^2(\omega t) 
		\biggl[
			&\innerp[\Big]{ \psi_t }{ (\spin-\spin+\otimes\an\ad) \psi_t }
			\\
 			&+ \innerp[\Big]{ \psi_t }{ (\spin+\spin-\otimes\ad\an) \psi_t }
			\\%
			&+ 2\Re\of*{\mrm{e} ^ {+\iu2\omega t}\innerp[\Big]{ (\spin-\otimes\an) \psi_t }{(\spin+\otimes\ad) \psi_t }}
			\biggr],
		\end{aligned}
		\raisetag{7em}
		\intertext{
			bound \(\Re(z)\leq \abs{z} = \sqrt{\Re(z)^2+\Im(z)^2}\), and separate the spin operators from the bosonic operators in the last line:
		}
			&\leq\begin{aligned}[t]
			\frac{\lambda^2}{\omega^2}\sin^2(\omega t)
			\biggl[
			&\innerp[\Big]{ \psi_t }{ \of*{\spin-\spin+\otimes\of{\ad\an+\one}} \psi_t }
			\\
 			&+ \innerp[\Big]{ \psi_t }{ \of*{\spin+\spin-\otimes\ad\an} \psi_t }
 			\\
			&+ 2\abs*{\innerp[\Big]{\of*{(\spin-)^2\otimes\one} \psi_t }{\of*{\one\otimes(\ad)^2} \psi_t }}
			\biggr],
            \label{eq:dicke-tc-general-bound-proof-SpSm-choice}
		\end{aligned}
		\intertext{
			use the Cauchy--Schwarz inequality and the positivity of \(\spin+\spin-\):
		}
			&\leq\begin{aligned}[t]
			\frac{\lambda^2}{\omega^2}\sin^2(\omega t)	
			\biggl[
			&\innerp[\Big]{ \psi_t }{ \of*{\spin-\spin+\otimes\of{\ad\an+\one}} \psi_t }
			\\
 			&+ \innerp[\Big]{ \psi_t }{ (\spin+\spin-\otimes\ad\an) \psi_t }
 			+ \innerp[\Big]{\psi_t }{(\spin+\spin-\otimes\one) \psi_t }
 			\\
			&+ 2 \norm[\Big]{\of*{(\spin-)^2\otimes\one} \psi_t }\, \norm[\Big]{\of*{\one\otimes(\ad)^2} \psi_t }
			\biggr],
		\end{aligned}
		\intertext{
		      collect the first three summands and apply \cref{eq:bosonic-monomial-bound} in the last summand:
		}
			&\leq\begin{aligned}[t]
			\frac{\lambda^2}{\omega^2}\sin^2(\omega t)
			\biggl[
			&\innerp[\Big]{ \psi_t }{\of*{\of{\spin-\spin+ + \spin+\spin-}\otimes\of{\ad\an+\one}} \psi_t }
 			\\
			&+ 2 \norm[\Big]{\of*{(\spin-)^2\otimes\one} \psi_t }\, \norm[\Big]{\of[\big]{\one\otimes\sqrt{(\nn+1)(\nn+2)}} \psi_t }
			\biggr],
		\end{aligned}
		\raisetag{4.6em}
		\intertext{
            use \cref{eq:SplusSminusCommutator} and \(\nn+1\leq\nn+2\):
		}
			&\leq\begin{aligned}[t]
			\frac{\lambda^2}{\omega^2}\sin^2(\omega t)
			\biggl[
			&\innerp[\Big]{\psi_t }{\of*{2\of*{{\spin{}}^2 - \Sz^2}\otimes\of{\ad\an +\one}}\psi_t }
			\\
			&+ 2 \norm[\Big]{\of*{(\spin-)^2\otimes\one} \psi_t }\, \norm[\Big]{\of[\big]{\one\otimes (\nn+2)} \psi_t }
			\biggr],
		\end{aligned}
        \intertext{
            use \({\spin{}}^2 - \Sz^2 \geq 0\) and \(\nn + \one \geq 0\):
		}
			&=\begin{aligned}[t]
			\frac{\lambda^2}{\omega^2}\sin^2(\omega t)
			\biggl[
			&\abs*{\innerp[\Big]{2\of*{{\spin{}}^2 - \Sz^2}\otimes \one \psi_t }{\one\otimes\of{\ad\an +\one}\psi_t }}
			\\
			&+ 2 \norm[\Big]{\of*{(\spin-)^2\otimes\one} \psi_t }\, \norm[\Big]{\of[\big]{\one\otimes (\nn+2)} \psi_t }
			\biggr],
		\end{aligned}
		\intertext{
			and finally use the Cauchy-Schwarz inequality:
		}
			&\leq\begin{aligned}[t]
			\frac{\lambda^2}{\omega^2}\sin^2(\omega t)
			\biggl[
			&2\norm[\Big]{\of*{{\spin{}}^2 - \Sz^2}\otimes \one \psi_t} \norm[\Big]{\one\otimes\of{\nn +1}\psi_t}
			\\
			&+ 2 \norm[\Big]{\of*{(\spin-)^2\otimes\one} \psi_t }\, \norm[\Big]{\of[\big]{\one\otimes (\nn+2)} \psi_t }
			\biggr],
		\end{aligned}
        \\
        &\leq\begin{aligned}[t]
			\frac{\lambda^2}{\omega^2}\sin^2(\omega t)
			\biggl[
			&2\norm[\Big]{\of*{{\spin{}}^2 - \Sz^2}\otimes \one \psi_t}
			\\
			&+ 2 \norm[\Big]{\of*{(\spin-)^2\otimes\one} \psi_t }
			\biggr]\norm[\Big]{\of[\big]{\one\otimes (\nn+2)} \psi_t }
		\end{aligned}
        \\
        &= \frac{\lambda^2}{\omega^2}\sin^2(\omega t)\of*{f^{\pm}_C(\psi_t)}^2
            .
		\end{align}    
  Note that in the last term of \cref{eq:dicke-tc-general-bound-proof-SpSm-choice}, we chose to move the spin operators to the left side of the inner product and the bosonic ones to the right. This leads to the appearance of \(\norm{\of{\spinalt-}^2\psi}\) in the final bound. Alternatively, one could reverse this separation, placing the bosonic operators to the left and the spin operators to the right. Since the quadratic bosonic term is in both cases estimated using the worst-case bound \cref{eq:bosonic-monomial-bound}, the only difference in the final expression is that \(\norm{\of{\spinalt+}^2\psi}\) would appear instead.
Taking the square root and minimizing over the two possible choices yields	
\begin{equation}\label{eq:dicke-tc-general-bound-proof5}
		\norm[\Big]{S_{21}(t)U_2(t)\psi} 
		\leq \frac{\lambda}{\omega}\abs*{\sin(\omega t)}\min_\pm f^{\pm}_C(\psi_t)
            .
	\end{equation}
	In the second addend of \cref{eq:dicke-tc-general-bound-proof1} we rewrite 
	\begin{align}
		S_{21}(t)H_2 - H_1(t)S_{21}(t) = \comm{S_{21}(t)}{H_2} + \of[\big]{H_2-H_1(t)}S_{21}(t)
        .
		\label{eq:general-bound-proof2}
	\end{align}
	Using the CCR (\cref{eq:CCR}) and the ladder algebra (\cref{eq:ladder-algebra}) respectively, we get
	\begin{align}
		\comm{S_{21}(t)}{H_2} 
		&= -\frac{\lambda^2}{\omega}\sin(\omega t) \comm[\Big]{
			\mrm{e} ^ {+\iu\omega t} \spin+\otimes\ad + \mrm{e} ^ {-\iu\omega t}\spin-\otimes\an 
		}{
			\spin+\otimes\an + \spin-\otimes\ad
		}
		\\
		&=\begin{aligned}[t]
			-\frac{\lambda^2}{\omega}\sin(\omega t) 
			\Bigl(
			&\mrm{e} ^ {+\iu\omega t}\of[\big]{
				\of{\spin+}^2\otimes\comm{\ad}{\an} + \comm{\spin+}{\spin-}\otimes\of{\ad}^2
			} 
			\\
			&+ \mrm{e} ^ {-\iu\omega t}\of[\big]{
			 \comm{\spin-}{\spin+}\otimes \an^2  + \of{\spin-}^2\otimes\comm{\an}{\ad}
			}
			\Bigr)
		\end{aligned}
		\\
		&=\begin{aligned}[t]
		 -\frac{\lambda^2}{\omega}\sin(\omega t) 
			\Bigl(
			&\mrm{e} ^ {+\iu\omega t}\of[\big]{
			-\of{\spin+}^2\otimes\one 
			+ 2\SZ\otimes\of{\ad}^2 
			}
			\\
			&+ \mrm{e} ^ {-\iu\omega t}\of[\big]{
			 -2\SZ\otimes \an^2  
			+ \of{\spin-}^2\otimes\one
			}
			\Bigr)
		\end{aligned}
		\\
		&=\begin{aligned}[t]
		-\frac{\lambda^2}{\omega}\sin(\omega t)  
			\Bigl(
			&\s+\s- \otimes \of*{ \mrm{e} ^ {+\iu\omega t}\of{\ad}^2 - \mrm{e} ^ {-\iu\omega t} \an^2 }
			\\
			&-\s-\s+ \otimes \of*{ \mrm{e} ^ {+\iu\omega t}\of{\ad}^2 - \mrm{e} ^ {-\iu\omega t} \an^2 }
			\\
			&+ \mrm{e} ^ {-\iu\omega t} \of[\big]{\s-}^2\otimes\one 
			- \mrm{e} ^ {+\iu\omega t}\of[\big]{\s+}^2\otimes\one
			\Bigl).
		\end{aligned}
	\end{align}
	There, we expanded \(2\SZ=\s+\s- - \s-\s+\) since, for the second addend in \cref{eq:general-bound-proof2}, by \cref{eq:H1,eq:H2,eq:integrated-action} we get
	\begin{align}
		\MoveEqLeft
		\of{H_2-H_1(t)}S_{21}(t) 
		\\
		&\!\begin{multlined}[t]
			=-\frac{\lambda^2}{\omega}\sin(\omega t)\of{
			-\mrm{e} ^ {+\iu2\omega t} \s+\otimes\ad - \mrm{e} ^ {-\iu2\omega t} \s-\otimes\an }
		\\
			\cdot \of{ \mrm{e} ^ {+\iu\omega t} \s+\otimes\ad + \mrm{e} ^ {-\iu\omega t} \s-\otimes\an}
		\end{multlined}
		\\
		&\!\begin{multlined}[t]
		=-\frac{\lambda^2}{\omega}\sin(\omega t)
			\Bigl[
				-\mrm{e} ^ {+\iu 3\omega t}\of[\big]{\s+}^2\otimes\of{\ad}^2
				-\mrm{e} ^ {+\iu \omega t}\s+\s-\otimes\ad\an
		\\
				-\mrm{e} ^ {-\iu \omega t}\s-\s+\otimes\an\ad
				-\mrm{e} ^ {-\iu 3\omega t}\of[\big]{\s-}^2\otimes \an^2
			\Bigr],
		\end{multlined}
	\end{align}
	so that we can collect terms by the respective product of \(\s+\) and \(\s-\).
	By doing so, we obtain
	\begin{align}
	\begin{aligned}[t]
		S_{21}(t)H_2 - H_1(t)S_{21}(t)
		&= -\frac{\lambda^2}{\omega}\sin(\omega t)
		\begin{multlined}[t]
		\Bigl[
		\s+\s- \otimes f^{(2)}_3(\an,\ad)
		-\s-\s+\otimes \tilde f^{(2)}_3(\an,\ad)
		\\
		- \of[\big]{\s+}^2\otimes g^{(2)}_2(\an,\ad)
		+ \of[\big]{\s-}^2\otimes h_2^{(2)}(\an,\ad)
		\Bigr],
		\end{multlined}
	\end{aligned}
	\raisetag{3.8em}
	\end{align}
	where
	\begin{align}
		f^{(2)}_3(\an,\ad) &= \of[\big]{
			\mrm{e} ^ {+\iu\omega t}(\ad)^2 
			- \mrm{e} ^ {-\iu\omega t}\an^2
			- \mrm{e} ^ {+\iu\omega t}\ad\an
		}\\
        \tilde f^{(2)}_3(\an,\ad) &= \of[\big]{
			\mrm{e} ^ {+\iu\omega t}(\ad)^2 
			- \mrm{e} ^ {-\iu\omega t}\an^2
			- \mrm{e} ^ {-\iu\omega t}\an\ad
		}\\
		g^{(2)}_3(\an,\ad) &= \of[\big]{
			\mrm{e} ^ {+\iu\omega t}\an\ad - \mrm{e} ^ {+\iu\omega t}\ad\an + \mrm{e} ^ {+\iu3\omega t}(\ad)^2
		}\\
		h_3^{(2)}(\an,\ad) &= \of[\big]{
			\mrm{e} ^ {-\iu\omega t}\an\ad - \mrm{e} ^ {-\iu\omega t}\ad\an - \mrm{e} ^ {-\iu3\omega t}\an^2
		}.
	\end{align}
	by \cref{eq:general-bound-proof2}.
	In the norm, we can now use the triangle inequality and the \textit{Cauchy--Schwarz trick}~\eqref{eq:Cauchy--Schwarz-trick} to separate the collective spin operators and the field operators.
	By doing this we get
	\begin{align}
		&\norm*{\of[\big]{S_{21}(t)H_2 - H_1(t)S_{21}(t)} \psi_t}
		\\
		&\quad\leq\begin{aligned}[t]
		\frac{\lambda^2}{\omega}
		\biggl[
			&\sqrt{
			\norm[\big]{\of[\big]{\salt+\salt-}^2 \psi_t}
			\norm[\big]{\one\otimes\of[\big]{f^{(2)}_3(\an,\ad)}\adj f^{(2)}_3(\an,\ad) \psi_t}
			}
			\\
			&+ \sqrt{
			\norm[\big]{\of[\big]{\salt-\salt+}^2\psi_t}
			\norm[\big]{\one\otimes\of[\big]{\tilde f^{(2)}_3(\an,\ad)}\adj \tilde f^{(2)}_3(\an,\ad) \psi_t}
			}
			\\
			&+ \sqrt{
			\norm[\big]{ \of[\big]{\salt-}^2 \of[\big]{\salt+}^2 \psi_t }
			\norm[\big]{ \one\otimes\of[\big]{g^{(2)}_3(\an,\ad)}\adj g^{(2)}_3(\an,\ad) \psi_t }
			}
			\\
			&+ \sqrt{
			\norm[\big]{ \of[\big]{\salt-}^2 \of[\big]{\salt+}^2 \psi_t }
			\norm[\big]{ \one\otimes\of[\big]{h_3^{(2)}(\an,\ad)}\adj h_3^{(2)}(\an,\ad) \psi_t }
			}
		\biggr].
	\end{aligned}\label{eq:dicke-tc-general-bound-proof3}
	\end{align}
    Only now we apply the bound for bosonic polynomials (\cref{thm:bosonic-polynomial-bound}) where in \(f^{(2)}_3\), \(g^{(2)}_3\), and \(h^{(2)}_3\) all coefficients are complex phases with \(\abs{\lambda_i} = 1\) and only the bosonic side of the system is considered. For example,
	the product 
     \begin{equation}
        \of[\big]{f^{(2)}_3(\an,\ad)}\adj f^{(2)}_3(\an,\ad)
     \end{equation}
    is a polynomial of \(3\cdot3=9\) addends of order up to \(2+ 2=4\).
	Therefore, \cref{thm:bosonic-polynomial-bound} gives us
	\begin{align}
	\MoveEqLeft
		\norm[\Big]{\one\otimes\of[\big]{f^{(2)}_3(\an,\ad)}\adj f^{(2)}_3(\an,\ad) \psi_t} 
		\\
		&\leq 9\norm*{\sqrt{(\nnalt+1)(\nnalt+2)(\nnalt+3)(\nnalt+4)} \psi_t}
		\\
		&\leq 9 \norm*{\sqrt{(\nnalt+4)(\nnalt+4)(\nnalt+4)(\nnalt+4)} \psi_t}
		\\
		&= 9 \norm*{(\nnalt+4)^2 \psi_t}.
	\end{align}
	Similarly, the other products of bosonic polynomials are also bounded by \(9 \norm{(\nnalt+4)^2 \psi_t}\), as they consist of \(9\) addends of order \(4\).
	Therefore, in \cref{eq:dicke-tc-general-bound-proof3} we get
	\begin{align}
	\MoveEqLeft
		\norm*{\of[\big]{S_{21}(t)H_2 - H_1(t)S_{21}(t)}\psi_t}
		\\
		&\leq 
		\begin{aligned}[t] 
		3\frac{\lambda^2}{\omega}
		\biggl[
			&\norm[\big]{ \of{\salt-}^2 \of{\salt+}^2\psi_t } ^{\nicefrac{1}{2}}
			+\norm[\big]{ \of{\salt-}^2 \of{\salt+}^2\psi_t } ^{\nicefrac{1}{2}}
			\\
			&+ \norm[\big]{\of{\salt+\salt-}^2\psi_t} ^{\nicefrac{1}{2}}
			+ \norm[\big]{\of{\salt-\salt+}^2\psi_t}^{\nicefrac{1}{2}}
		\biggr]
		\cdot\sqrt{\norm{(\nnalt+4)^2 \psi_t}},
		\end{aligned}
	\raisetag{3.75em}
    \\
    &= 3\frac{\lambda^2}{\omega} f_L(\psi_t)
	\label{eq:dicke-tc-general-bound-proof4}
	\end{align}
	Combining \cref{eq:dicke-tc-general-bound-proof1,eq:dicke-tc-general-bound-proof4,eq:dicke-tc-general-bound-proof5} gives \cref{eq:intermediate-bound} and thus concludes the proof.
\end{proof}

\begin{proof}[{Proof of \cref{thm:scaling-S-M-bound}}]\label{proof:scaling-S-M-bound}
We express \(\tcpsi{S}{M}{\sigma}\) in terms of Dicke states \(\ket{S,m}\) and Fock states \(\phi_n\) using \cref{eq:TC-eigenstates}:
    \begin{align}
        \tcpsi{S}{M}{\sigma}
        &= \prod_{i=1}^M \of*{\tcpar{i}{\sigma}\adalt - \salt+} \ket{S,-S}\otimes \phi_0
        \\
        &= \sum_{Z \subset \Set{\tcpar{i}{\sigma}}_{i=1}^{i=M}} (-1)^{M - \abs{Z}} \of[\Big]{\prod_{z\in Z}z} \of[\big]{\adalt}^{\abs{Z}} \of[\big]{\salt+}^{M-\abs{Z}} \ket{S,-S}\otimes \phi_0
        \\
        \label{eq:tcpsi-expansion}
        &= \sum_{n=M-2S}^{M} c_{S,M,n}^\sigma \ket{S, M-n-S}\otimes\phi_n
    \end{align}
    where 
    \begin{align}
        \label{eq:tc-eigenstates-coefficients}
        c_{S,M,n}^\sigma =  (-1)^{M-n} \sqrt{\frac{n!(M-n)!(2S)!}{(2S - M +n)!}} \sum_{\substack{Z \subset \Set{\tcpar{i}{\sigma}}_{i=1}^{i=M}\\ \abs{Z}=n}}\prod_{z\in Z}z
    \end{align}
    and \(\abs{Z}\) is the number of elements in the subset \(Z\).
    Here we used the following property:
    \begin{align}
        \label{eq:Spm-dicke-states}
        \spin\pm \ket{S,m} 
        &= \begin{cases}
            C^\pm_{S,m} \ket{S, m\pm 1}, & \abs{m\pm 1} \leq S
            \\
            0,& \text{else}
        \end{cases}
        \shortintertext{where}
        \label{eq:def-CpmSm}
        C^\pm_{S,m} &= \sqrt{\of{S\mp m} \of{S\pm m +1}},
        \shortintertext{from which we obtain}
        \of[\big]{\spin+}^n \ket{S,-S} 
        &= \prod_{k=0}^{n-1} C^+_{S,-S+k}\ket{S, -S + n}
        \\
        &= \of*{ \prod_{k=0}^{n-1} (2S-k)(k+1) }^{\nicefrac{1}{2}} \ket{S, -S + n}
        \\
        &= \of*{ \prod_{k=0}^{n-1} (2S-k)}^{\nicefrac{1}{2}} \sqrt{n!} \ket{S, -S + n}
        \\
        &= \sqrt{\frac{(2S)!\, n!}{(2S-n)!}}  \ket{S, -S + n}
    \end{align}
    for \(n\leq 2S\) and \(\of[\big]{\spin+}^n \ket{S,-S} = 0\) for \(n>2S\).
    As the summands in \cref{eq:tcpsi-expansion} are mutually orthogonal, we immediately obtain
    \begin{gather}
        \label{eq:tc-eigenstates-norm}
        \norm{\tcpsi{S}{M}{\sigma}} 
        = \of*{\sum_{n = M-2S}^M \abs{c_{S,M,n}^\sigma}^2}^{\nicefrac{1}{2}}
        .
    \end{gather}
    
    \Cref{eq:tcpsi-expansion} also allows us to evaluate the remaining terms in the bound~\eqref{eq:general-bound} when taking into account 
    \begin{gather}
    \psi 
    = \tcpsinormed{S}{M}{\sigma}
    =\frac{\tcpsi{S}{M}{\sigma}}{\norm[\big]{\tcpsi{S}{M}{\sigma}}}
    = \sum_{n=M-2S}^{M} \tilde c_{S,M,n}^\sigma \ket{S, M-n-S}\otimes\phi_n
    ,
    \\ \\
    \tilde c_{S,M,n}^\sigma = \frac{c_{S,M,n}^\sigma}{\of[\Big]{\sum_{n = M-2S}^M \abs{c_{S,M,n}^\sigma}^2}^{\nicefrac{1}{2}}}
    .
    \end{gather}
    For example, we have
    \begin{align}
    	\norm{\salt+\salt+ \tcpsinormed{S}{M}{\sigma}}
    	&= \norm*{\sum_{n=M-2S}^{M} \tilde c_{S,M,n} ^\sigma \salt+\salt+\ket{S, M-n-S}\otimes\phi_n}
    	\\
        \label{eq:scaling-S-M-proof-1}
    	&= \left(\sum_{n = M-2S+2}^{M}\right. \underbrace{\abs{\tilde c_{S,M,n} ^\sigma}^2}_{\leq 1}\left.\vphantom{\sum_{n = M-2S+2}^{M}} \of*{C^+_{S,M-n-S+1} C^+_{S,M-n-S}}^2 \right)^{\nicefrac{1}{2}}
        \\
        &\leq \of*{\sum_{n = M-2S+2}^{M} \of*{C^+_{S,M-n-S+1} C^+_{S,M-n-S}}^2 }^{\nicefrac{1}{2}}
        \\
        &\!\begin{multlined}[t]
            = \left(\sum_{n = M-2S+2}^{M} \Bigl(\of{S- (M-n-S+1)} \of{S+ (M-n-S+1) +1}\cdot\right.
            \\
            \left.\cdot\of{S- (M-n-S)} \of{S+ (M-n-S) +1}\Bigr) \vphantom{\sum_{n = M-2S+2}^{M}}\right)^{\nicefrac{1}{2}}
        \end{multlined}
        \\
        &\leq \of*{\sum_{n = M-2S+2}^{M} \of*{n-(M-2S)}^2 \of*{M+2-n}^2 } ^{\nicefrac{1}{2}}
        \\
        &= \of*{\sum_{n = 0}^{2S-2} \of*{n+2}^2 \of*{2S-n}^2 } ^{\nicefrac{1}{2}}
        \\
        &= \sqrt{\frac{2}{15}}\of*{8 S^5+20 S^4+10 S^3-5 S^2-3 S}^{\nicefrac{1}{2}}
    	.
    \end{align}
    Analogously, we get
    \begin{align}
    	\norm*{\salt-\salt- \tcpsinormed{S}{M}{\sigma}}
    	&\leq \of*{\sum_{n = M-2S}^{M-2} C^-_{S,M-n-S-2} C^-_{S,M-n-S-1} }^{\nicefrac{1}{2}}
        \\
        &\leq \frac{2}{\sqrt{15}}\of*{8 S^5+20 S^4+10 S^3-5 S^2-3 S}^{\nicefrac{1}{2}}
    	,
        \\
    	\norm*{\of*{{\spinalt{}}^2 - \boldsymbol{S}_z^2} \tcpsinormed{S}{M}{\sigma}}
    	&\leq \of*{\sum_{n = M-2S}^{M}  \of*{S(S+1)-(M-n-S)^2}^2}^{\nicefrac{1}{2}}
        \\
        &= \sqrt{\frac{1}{15
        }}\of*{16 S^5+40 S^4+30 S^3+5 S^2-S}^{\nicefrac{1}{2}}
        ,
        \\
    	\norm*{\of*{\salt+\salt-}^2 \tcpsinormed{S}{M}{\sigma}}
    	&\leq \of*{\sum_{n = M-2S}^{M-1}  \of*{M-n}^4\of*{2S - M+n+1}^4}^{\nicefrac{1}{2}}
        \\
        &\!\begin{multlined}[t]
            = \frac{4}{{3 \sqrt{35}}} \bigl(16 S^9+72 S^8+144 S^7 +168 S^6+126 S^5
            \\
            +63 S^4+26 S^3+12 S^2+3S\bigr)^{\nicefrac{1}{2}}
        \end{multlined}
    	,
        \\
    	\norm*{\of*{\salt+}^2\of*{\salt-}^2 \tcpsinormed{S}{M}{\sigma}}
    	&\leq \of*{\sum_{n = M-2S}^{M-2}  \of*{M-n}^4\of*{2S-M+n+2}^4}^{\nicefrac{1}{2}}
        \\
        &\!\begin{multlined}[t]
            = \frac{1}{3\sqrt{35}}\bigl(256 S^9+2304 S^8+9216 S^7+21504 S^6+32256 S^5
            \\
            +22176 S^4+1424S^3-5664 S^2-2517 S-315\bigr)^{\nicefrac{1}{2}}
        \end{multlined}
    	,
        \\
    	\norm*{\of*{\salt-\salt+}^2 \tcpsinormed{S}{M}{\sigma}}
    	&\leq \of*{\sum_{n = M-2S+1}^{M} \of*{M-n+1}^4\of*{2S - M+n}^4}^{\nicefrac{1}{2}}
        \\
        &\!\begin{multlined}[t]
            = \frac{4}{3 \sqrt{35}} \bigl(16 S^9+72 S^8+144 S^7+168 S^6+126 S^5
            \\
            +63 S^4+26 S^3+12 S^2+3
            S\bigr)^{\nicefrac{1}{2}}
        \end{multlined}
        ,
        \\
    	\norm*{\of*{\salt-}^2\of*{\salt+}^2 \tcpsinormed{S}{M}{\sigma}}
    	&\leq \of*{\sum_{n = M-2S+2}^{M}  \of*{M-n+2}^4 \of*{2S-M+n+1}^4}^{\nicefrac{1}{2}}
        \\
        &\!\begin{multlined}[t]
            =\frac{4}{3\sqrt{35}} \bigl(16 S^9+216 S^8+1296 S^7+4536 S^6+10206 S^5
            \\
            +9639 S^4+2714S^3-1476 S^2-1317 S-315\bigr)^{\nicefrac{1}{2}}
        \end{multlined}
        ,
        \\
        \norm[\big]{\of{\nnalt + 2 } \tcpsinormed{S}{M}{\sigma}}
        &\leq \of*{\sum_{n = M-2S}^{M} (n+2)^2}^{\nicefrac{1}{2}}
        \\
        &\!\begin{multlined}[t]
            =\frac{1}{\sqrt{3}}\bigl(8S^3 +3M^2 +6M^2S -12MS^2 +18MS  
            \\
            +12M -18S^2 +13S +12\bigr)^{\nicefrac{1}{2}}
        \end{multlined}
        ,
        \shortintertext{and}
        \norm[\big]{\of{\nnalt + 4}^2 \tcpsinormed{S}{M}{\sigma}}
        &\leq \of*{\sum_{n = M-2S}^{M} (n+4)^4}^{\nicefrac{1}{2}}
        \\
        &\!\begin{multlined}[t]
            = \frac{1}{\sqrt{15}}\bigl(
            96S^5 +15M^4 +30M^4S -120M^3S^2 +420M^3S 
            \\
            +240M^2S^3 -1260M^2S^2 +2190M^2S +240M^3 +1440M^2
            \\
            -240MS^4 +1680MS^3 -4380MS^2 +5040MS +3840M
            \\
             -840S^4 +2920S^3 -5040S^2 +4319S +3840
            \bigr)^{\nicefrac{1}{2}}
        \end{multlined}
        ,
    \end{align}
    where we also used
    \begin{align}
    	C^+_{S,m-1}C^-_{S,m} 
    	&= \of*{C^-_{S,m}}^2
    	.
    \end{align}
    Inserting these estimates into \cref{eq:def-fC-fL} and simplifying the resulting expression gives
    \begin{align}
        \label{eq:f-C-explicit-eval}
        f_C^\pm\of*{\tcpsi{S}{M}{\sigma}}
        &\leq\!\begin{multlined}[t] 
           \frac{\sqrt{2}}{\sqrt[4]{45}}\of*{(2 S+1) \left(3 M^2-6 M (S-2)+4 S^2-11S+12\right)}^{\nicefrac{1}{4}}
            \cdot
           \\
           \cdot
            \left(
            \sqrt{2} \sqrt{S \left(8 S^4+20 S^3+10 S^2-5 S-3\right)}
            \right.
            \\
            \left.
            +\sqrt{S\left(16 S^4+40 S^3+30 S^2+5 S-1\right)}
            \right)^{\nicefrac{1}{2}}
        \end{multlined}
        \shortintertext{and}
        \label{eq:f-L-explicit-eval}
        f_L\of*{\tcpsi{S}{M}{\sigma}}
        &\leq\!\begin{multlined}[t] 
            \frac{2}{\sqrt[4]{4725}}
            \Bigg(
                2 \left(
                    16 S^9+72 S^8+144 S^7+168 S^6+126 S^5+63 S^4+26 S^3+12S^2
                \right.
                \\
                \left.
                +3S \vphantom{12S^2}
                \right)^{\nicefrac{1}{4}}
                +2^{3/4} 
                \left(
                    32 S^9+144 S^8+240 S^7+168 S^6+42 S^5+21 S^4
                \right.
                \\
                \left.
                    +10S^3-18 S^2-9S
                \right)^{\nicefrac{1}{4}}
            \Bigg)
                (2 S+1)^{\nicefrac{1}{4}}
            \Bigg(
                    15 M^4-60 M^3 (S-4)+1682S^2
                \\
                    +30 M^2 
                    \left(
                        4S^2-23 S+48
                    \right)
                    -60 M \left(
                        2 S^3-15 S^2+44 S-64
                    \right)
            \\
                    +48 S^4-444 S^3
                    -3361 S+3840
            \Bigg)^{\nicefrac{1}{4}}
            .
        \end{multlined}
    \end{align}
    By considering the highest order terms in each of these expressions, we infer 
    \begin{align}
        f_C^{\pm}\of*{\tcpsi{S}{M}{\sigma}} 
        &= \mathcal{O}\of*{S^2 M^{\nicefrac{1}{2}}}
        \\
        f_L\of*{\tcpsi{S}{M}{\sigma}}
        &= \mathcal{O}\of*{S^{\nicefrac{7}{2}} M}
        .
    \end{align}
    This means there exist \(a,b,c,d,e,f \in \R\) s.t. 
    \begin{align}
        f_C^{\pm}\of*{\tcpsi{S}{M}{\sigma}}
        \leq a (S+b)^2 (M+c)^{\nicefrac{1}{2}}
        \label{eq:scaling-S-M-proof-3}
        \\
        f_l\of*{\tcpsi{S}{M}{\sigma}}
        \leq d (S+e)^{\nicefrac{7}{2}} (M+f)
        \label{eq:scaling-S-M-proof-4}
        .
    \end{align}
   One possible choice of such coefficients is \(a=2,\, b=1,\, c=2,\, d=6,\, e=1,\, f=2\).
Inserting \cref{eq:scaling-S-M-proof-3,eq:scaling-S-M-proof-4} into \cref{eq:general-bound} with this choice of parameters yields \cref{eq:scaling-S-M-bound}.
\end{proof}

\setlength{\emergencystretch}{1em}
\bibliographystyle{myquantum}
\bibliography{CONTENT/references}
\end{document}